\documentclass[acmsmall, screen, nonacm]{acmart}

\usepackage{mathtools}
\usepackage{graphicx}
\usepackage{xcolor}
\usepackage{mathpartir}
\usepackage{dashbox}
\usepackage{hyperref}
\usepackage{cleveref}
\usepackage{tikz}
\usetikzlibrary{
  arrows.meta,
  automata,
  positioning,
  shapes,
  patterns,
  arrows.meta,
  calc
}
\usepackage{fancybox}
\usepackage{subcaption}
\usepackage{listings}
\usepackage{paralist}
\usepackage{todonotes}
\usepackage{ifthen}
\newboolean{fullversion}
\setboolean{fullversion}{true} % true: arxiv version, false: final version
\iffalse
\newcommand\arthur[1]{\textcolor{red}{Arthur: #1}}
\newcommand\jana[1]{\textcolor{blue}{Jana: #1}}
\newcommand\haoyi[1]{\textcolor{green}{Haoyi: #1}}
\else
\newcommand\arthur[1]{}
\newcommand\jana[1]{}
\newcommand\haoyi[1]{}
\fi

\definecolor{stepcol}{rgb}{0.0, 0.53, 0.74}
\definecolor{obscol}{rgb}{0.0, 0.53, 0.74}
\definecolor{hcol}{RGB}{197, 40, 61}

\newcommand{\set}[1]{\{#1\}}
\newcommand{\tup}[1]{( #1 )}
\newcommand{\ldot}{\mathpunct{.}}
\newcommand{\Nat}{\mathbb{N}}
\newcommand{\Int}{\mathbb{Z}}
\newcommand{\cons}{\hspace{-2pt}::\hspace{-2pt}}

\newcommand{\transsys}{\mathcal{T}}

\newcommand{\hw}{\mathcal{H}}
\newcommand{\sw}{\mathcal{C}}

\newcommand{\lstep}[1]{{\color{stepcol}\xrightarrow{{\color{obscol}#1}}_\sw} \,}

\newcommand{\hweval}[1]{{\color{hcol}\llbracket} #1 {\color{hcol}\rrbracket_{\hw}}}
\newcommand{\ceval}[1]{{\color{stepcol}\llbracket} #1 {\color{stepcol}\rrbracket_{\sw}}}
\newcommand{\hwobs}[3]{\hweval{#1}(#2, #3)}
\newcommand{\cobs}[2]{\ceval{#1}(#2)}
\newcommand{\hwsem}{\hweval{\hspace{-3pt} - \hspace{-3pt}}}
\newcommand{\csem}{\ceval{\hspace{-3pt} - \hspace{-3pt}}}

\newcommand{\kywd}[1]{\mathbf{#1}}

\newcommand{\loadKywd}{\kywd{load}}

\newcommand{\jzKywd}{\kywd{beqz}}

\newcommand{\addKywd}{\kywd{add}}

\newcommand{\pseq}[2]{#1 ; #2}

\newcommand{\tadd}[3]{\addKywd\ #1, #2, #3}

\newcommand{\pload}[2]{\loadKywd\ #1, #2}

\newcommand{\pjz}[2]{\jzKywd\ #1, #2}

\newcommand{\pc}{\mathit{pc}}

\newcommand{\ms}[1]{\ensuremath{\mathsf{#1}}}
\newcommand{\mi}[1]{\ensuremath{\mathit{#1}}}
\newcommand{\col}[2]{\ensuremath{{\color{#1}{#2}}}}
\newcommand{\archStyle}[1]{\ms{\col{stepcol}{#1}}}
\newcommand{\muarchStyle}[1]{\ms{\col{hcol}{#1}}}
\newcommand{\cstepstyle}[1]{{\mi{\col{stepcol}{#1}}}}
\newcommand{\hwstepstyle}[1]{{\mi{\col{hcol}{#1}}}}

\newcommand{\archsingleStep}[1]{\,\Rightarrow_\mi{isa}\,}
\newcommand{\hwsingleStep}[1]{~\muarchStyle{\Rightarrow^{\hspace{-7.5pt}\raisebox{2pt}{\scriptsize#1}}_\hw}~}

\newcommand{\amstep}[1]{~\archStyle{\xrightarrow{\cstepstyle{#1}}_\sw}~}
\newcommand{\hwstep}[1]{~\muarchStyle{\xrightarrow{\hwstepstyle{#1}}_\hw}~}

\newcommand{\trans}{\mathit{next}}
\newcommand{\leak}{\mathit{leak}}
\newcommand{\den}[1]{\llbracket #1 \rrbracket}

\newcommand{\gquadruple}[5]{\fbox{$#1$} \vdash \begin{bmatrix}
  #2 & #4\\
  #3 & #5
\end{bmatrix}}
\newcommand{\quadruple}[5]{\dbox{$#1$} \vdash \begin{bmatrix}
  #2 & #4\\
  #3 & #5
\end{bmatrix}}
\newcommand{\ngquadruple}[5]{#1 \vdash \begin{bmatrix}
  #2 & #4\\
  #3 & #5
\end{bmatrix}}

\newcommand{\sgquadruple}[5]{\fbox{$#1$} \vdash_\texttt{lockstep} \begin{bmatrix}
  #2 & #4\\
  #3 & #5
\end{bmatrix}}

\renewcommand\comment[1]{\textrm{\color{gray}(#1)}}
\newcommand\pcomment[1]{\textrm{\color{gray}(#1)}\quad}

\definecolor{Blue3}{HTML}{0000CD}
\definecolor{Green4}{HTML}{008B00}
\definecolor{Red3}{HTML}{CD0000}
\definecolor{orange}{rgb}{0.8, 0.47, 0.196}
\lstdefinestyle{Cstyle}
{
	frame = tb,
  belowskip=.4\baselineskip,
  aboveskip=.4\baselineskip,
  	showstringspaces = false,
  	breaklines = true,
  	breakatwhitespace = true,
  	tabsize = 3,
  	numbers = left,
    stepnumber = 1,
    numberstyle = \tiny\color{gray},
    language = {[ANSI]C},
    alsoletter={.\$},
    basicstyle={\ttfamily\color{black}},
    keywordstyle={\ttfamily\color{Green4}},
	morecomment=[l][\small\itshape\color{purple!40!black}]{//},
	sensitive=true,
}

\theoremstyle{plain}
\newtheorem{theorem}{Theorem}

\theoremstyle{definition}
\newtheorem{definition}[theorem]{Definition}

\AtBeginEnvironment{example}{%
  \pushQED{\qed}%
}
\AtEndEnvironment{example}{\popQED\endexample}

\begin{document}

\author{Arthur Correnson}
\orcid{0000-0003-2307-2296}
\affiliation{%
  \institution{CISPA Helmholtz Center for Information Security}
  \city{Saarbruecken}
  \country{Germany}
}
\email{arthur.correnson@cispa.de}

\author{Haoyi Zeng}
\orcid{0009-0007-2506-3787}
\affiliation{%
  \institution{Harvard University}
  \city{Cambridge, MA}
  \country{USA}
}
\email{haoyizeng@g.harvard.edu}

\author{Jana Hofmann}
\orcid{0000-0003-1660-2949}
\affiliation{
  \institution{Max Planck Institute for Security and Privacy (MPI-SP)}
  \city{Bochum}
  \country{Germany}
}
\email{jana.hofmann@mpi-sp.org}

% \makeatletter
% \everymath{%
%   \ifodd\thepage\allowdisplaybreaks[0]%
%     \else \allowdisplaybreaks[4]%
%   \fi
% }
% \makeatother

  \title{A Deductive System for Contract Satisfaction Proofs}

  \subtitle{Extended Version}

  \begin{abstract}
    Hardware-software contracts are abstract specifications of a CPU's leakage behavior.
    They enable verifying the security of high-level programs against side-channel attacks without having to explicitly reason about the microarchitectural details of the CPU.
    Using the abstraction powers of a contract requires proving that the targeted CPU \emph{satisfies} the contract in the sense that the contract
    over-approximates the CPU's leakage.
    Besides pen-and-paper reasoning, proving contract satisfaction has been approached mostly from the model-checking perspective, with approaches based on a (semi-)automated search for the necessary invariants.

    As an alternative, this paper explores how such proofs can be conducted in interactive proof assistants.
    We start by observing that contract satisfaction is an instance of a more general problem we call \emph{relative trace equality}, and we introduce \emph{relative bisimulation} as an associated proof technique.
    Leveraging recent advances in the field of coinductive proofs, we develop
    a deductive proof system for relative trace equality.
    Our system is provably sound and complete, and it enables a modular and incremental proof style.
    It also features several reasoning principles to simplify proofs by exploiting symmetries and transitivity properties.
    We formalized our deductive system in the Rocq proof assistant and applied it to two challenging contract satisfaction proofs.
  \end{abstract}

%   \begin{CCSXML}
% <ccs2012>
%    <concept>
%        <concept_id>10003752.10003790.10002990</concept_id>
%        <concept_desc>Theory of computation~Logic and verification</concept_desc>
%        <concept_significance>300</concept_significance>
%        </concept>
%    <concept>
%        <concept_id>10002978.10003001</concept_id>
%        <concept_desc>Security and privacy~Security in hardware</concept_desc>
%        <concept_significance>300</concept_significance>
%        </concept>
%    <concept>
%        <concept_id>10002978.10003006.10011608</concept_id>
%        <concept_desc>Security and privacy~Information flow control</concept_desc>
%        <concept_significance>300</concept_significance>
%        </concept>
%  </ccs2012>
% \end{CCSXML}


  \maketitle

  \section{Introduction}

  When multiple parties perform computations on shared hardware (e.g., in cloud computing), potentially secret data leaks into the microarchitecture and can be observed by other parties hosted on the same machine.
  The effectiveness of such side-channel attacks has been impressively demonstrated by the Meltdown~\cite{meltdown} and Spectre~\cite{spectre} attacks and the many exploits that followed~\cite{fallout, zombieload, ridl, foreshadow, lvi, inception, bhi}.
  Since disabling the sharing of resources is not an option for performance reasons, a lot of work has been invested on how to defend against such attacks.

  \paragraph{Hardware-Software Contracts}

  Developing low-impact software defenses requires a precise specification of the hardware's leakage behavior.
  Hardware-software contracts address this need by formally characterizing leakage at the level of the instruction set architecture (ISA).
  This enables software developers to assess a program's leakage without requiring deep knowledge of the microarchitectural details.
  For this approach to be correct, the contract needs to soundly over-approximate the leakage of the microarchitecture.
  This property, called \emph{contract satisfaction}, is usually expressed as follows~\cite{hwswcontracts}:
  
  \begin{quote}
    \textit{For any program $P$ and two inputs $\sigma, \sigma'$, if the contract predicts the same leakage when running $P$ on $\sigma$ and $\sigma'$, then the two runs are also indistinguishable on the hardware.}
  \end{quote}
  We make this contract satisfaction property more formal in \Cref{sec:HW-SW-contracts}.
  For now, note that the property is relational, i.e., a hyperproperty, as it reasons about two contract traces and two hardware traces.
  In this paper, we develop a deductive system dedicated to proving relational properties of this form.

  \paragraph{The Challenges of Proving Contract Satisfaction}

  So far, the task of validating that a given microarchitecture satisfies a contract has been tackled using testing~\cite{revizor, hideandseek, speculationatfault, amulet} and model checking~\cite{Leave,shadowlogic,upec}.
  Such automated techniques are often very effective, but they rely on the correctness of the underlying algorithms and their implementations.
  There also exist some rigorous pen-and-paper proofs of contract satisfaction, but their complexity and length make them difficult to verify.
  As an example, the technical report of the paper proposing hardware-software contracts provides a contract satisfaction proof for a simple processor featuring speculation and out-of-order execution. The corresponding pen-and-paper proof consists of 25 dense pages~\cite{hwswcontracts}.

  To achieve an additional level of formality, we set ourselves the goal of developing 
  contract satisfaction proofs inside a proof assistant (Rocq, in our case).
  This is challenging, as it involves 4-ary relational reasoning between execution traces induced by two different semantics.
  Additionally, executions progress \emph{asynchronously} (one contract execution step can correspond to several hardware steps and vice-versa), and relational invariants need to account for these potential misalignments.

  \paragraph{Existing Approaches}

  The challenges induced by this kind of 4-ary reasoning are not specific to hardware-software contracts.
  Similar issues arise whenever the security of two different systems is put in relation, for example, in relative security~\cite{DongolGPW24}, information-preserving refinement~\cite{AthalyeCKTZ24, knox}, and secure compilation~\cite{specconsttime, ctsim, ctjasmin, niopt, compcertsec, Wall2024SNIPSE, exorcisingspectres}.
  Proof techniques for such properties typically require guessing sophisticated witnesses up-front, often some form of simulation relation.
  In secure compilation, for example, a common technique defines three simulation relations: one between source states, one between target states, and one bridging the first two.
  This results in a three-dimensional simulation diagram justifying that equality of observations at the source level implies equality of observations at the target level \cite{specconsttime, ctsim, ctjasmin, niopt, compcertsec}.
  Coming up with such relations requires a deep understanding of the studied compiler.
  In the context of hardware-software contracts, there is no compiler between the high-level (the contract) and the low-level (the hardware).
  Instead, both sides execute the same program, expressed in the same ISA, but under two different semantics.
  In this specific setting, a more direct approach seems to be possible, one that constructs a single 4-ary relation on the fly during the proof.

  \paragraph{Our Approach}
  In this paper, we build such a proof technique for contract satisfaction.
  Specifically, we first observe that contract satisfaction is an instance of what we call \emph{relative trace equality}, i.e., properties of the form  ``trace equality implies trace equality''. 
  We then propose \emph{relative bisimulation} as a sound and complete technique to prove relative trace equality.
  Building on this definition, we present a deductive system for relative bisimulation proofs.
  Contrary to existing approaches, our system does not require providing simulation relations up-front.
  Instead, it takes advantage of recent advances in the field of coinductive proofs to construct invariants incrementally during the proof.
  To deal with the asynchronicity of the problem, our system supports non-lockstep reasoning via natural alternations of coinductive and inductive reasoning phases.

  We formalized our deductive system in Rocq and formally proved it sound and complete for relative trace equality.
  As a case study, we used our system to validate two standard contracts
  against two simple hardware models: an ``always-mispredict'' contract to reason about speculative execution of branches, and a ``sequential'' contract for out-of-order execution.
  Finally, we showed that our system supports several so-called \emph{up-to techniques} that exploit symmetry and transitivity properties that can simplify proofs.

  \paragraph{Structure of the Paper}

  The rest of the paper is structured as follows.
  Section 2 gives an overview of microarchitectural side channels and hardware-software contracts.
  Section 3 introduces the notion of proof by relative bisimulation, together with an associated deductive system.
  Section 4 demonstrates how to use our deductive system to validate two different hardware-software contracts.
  Finally, Section 5 presents several up-to techniques that can be soundly integrated into our deductive system.

  \section{Hardware-Software Contracts}
  \label{sec:HW-SW-contracts}

  Before introducing our proof technique, we recall the basics of microarchitectural side channels, noninterference, and hardware-software contracts.

  \subsection{Microarchitectural Side Channels}
  Computations performed on CPUs leave traces in the microarchitecture, e.g., in caches, buffers, and in the internal state of branch predictors. When several parties share microarchitectural components, one party can observe the microarchitectural changes generated by the others by using a so-called side channel.
  If these changes depend on secret data, this constitute a critical information leakage.
  In a cache-based side-channel attack, for example, the attacker can learn which addresses have been accessed by the victim (just the address, not the value stored at it). If the addresses depend on secrets, the attacker might be able to deduce the value of secret data.

  \subsection{Noninterference and Contract Satisfaction}

  \emph{Noninterference} establishes that the observations an attacker can make are independent of the secrets.
  Let $P$ be a program and $(\sigma, \mu)$ be a CPU state composed of an architectural part $\sigma$, and a microarchitectural part $\mu$. We write $\sigma =_\texttt{pub} \sigma'$ to denote that two architectural states agree on their public values (but might disagree on their secrets). Let furthermore $\hwobs{P}{\sigma}{\mu}$ denote the microarchitectural observations an attacker can make through a cache-based side channel when $P$ is executed with initial architectural state $\sigma$ and microarchitectural state $\mu$.
  $\hwobs{P}{\sigma}{\mu}$ is a (potentially infinite) trace typically defined in terms of 
  a labeled operational semantics of the hardware.
  The program $P$ is noninterferent on the modeled CPU if
  $$
    \forall \sigma, \sigma', \mu \ldot \sigma =_\texttt{pub} \sigma' \Rightarrow \hwobs{P}{\sigma}{\mu} = \hwobs{P}{\sigma'}{\mu} 
  $$
  To establish that a program is noninterferent, we have to reason about the microarchitectural behavior of the specific hardware it is run on.
  This requires deep knowledge of both the program and the microarchitecture.
  To split that task between hardware and software engineers, hardware-software contracts have been introduced as an abstraction layer between the two~\cite{hwswcontracts}.

  Contracts are formal models of the expected side-channel leakage of a CPU. They are defined in terms of the instruction set architecture (ISA) and thus abstract away the behavior of the microarchitecture. We write $\cobs{P}{\sigma}$ for the leakage predicted by a contract $\cstepstyle{C}$ when running $P$ on $\sigma$. As for the hardware, $\cobs{P}{\sigma}$ is a trace defined by a labeled operational semantics. Contracts let us split the noninterference property into two parts. The first property establishes that a hardware satisfies a contract.
   \begin{definition}[Contract Satisfaction~\cite{hwswcontracts, revizor}]
    \label{def:hwswcontracts}
    A hardware semantics $\hwsem$ satisfies a contract $\csem$ if
    $$
      \forall P, \sigma, \sigma', \mu \ldot \cobs{P}{\sigma} = \cobs{P}{\sigma'} \Rightarrow \hwobs{P}{\sigma}{\mu} = \hwobs{P}{\sigma'}{\mu}
    $$
  \end{definition}

  \noindent The second property establishes that a program is noninterferent w.r.t. a specific contract.
  \begin{definition}[\cite{hwswcontracts}]
    \label{def:contractnoninf}
    A program $P$ is noninterferent with respect to contract $\csem$ if
    $$
      \forall \sigma, \sigma' \ldot \sigma =_\texttt{pub} \sigma' \Rightarrow \cobs{P}{\sigma} = \cobs{P}{\sigma'}
    $$
  \end{definition}
  \noindent With these definitions, if a CPU with semantics $\hwsem$ satisfies a contract $\csem$, and $P$ is noninterferent w.r.t. $\csem$, then $P$ is also noninterferent on $\hwsem$. 
  The other direction does not generally hold, since $\csem$ usually over-approximates the leakage of $\hwsem$.
  This split into two properties makes reasoning about noninterference a lot more convenient: Definition~\ref{def:hwswcontracts} lets us reason about the leakage of a CPU without worrying about concrete programs and their secrets. Definition~\ref{def:contractnoninf} lets us reason about the leakages of a program without worrying about the concrete hardware semantics and its microarchitecture.
  In this paper, we focus on the contract satisfaction property given in Definition~\ref{def:hwswcontracts}.
  We define a deductive system to prove that a given hardware semantics satisfies a contract.

  \subsection{Example: Always-Mispredict Contract}
  \label{subsec:sec2-always-mispredict}
  The goal of developing a deductive system for contract satisfaction proofs was inspired by the dense pen-and-paper proofs that accompany the paper that first introduced this form of contracts~\cite{hwswcontracts}.
  The contracts developed in this paper use the always-mispredict feature for modeling Spectre-style branch misprediction and speculative execution~\cite{spectector}.
  We later use this feature to demonstrate the usage of our deductive system, so we introduce the intuition behind it.
  
  \paragraph{Branch Misprediction}
  Modern processors have a branch predictor, which predicts the outcome of branching decisions.
  If the result of evaluating the branch condition is not available yet, the CPU will speculatively execute the predicted branch and potentially roll back to the correct branch when the branch condition has finished evaluating.
  The architectural effects of speculative execution are fully rolled back, meaning that speculation does not impact functional correctness. Microarchitectural effects, however, cannot be rolled back.
  As a result, inadvertently executed branches might leak secrets into the microarchitecture.

  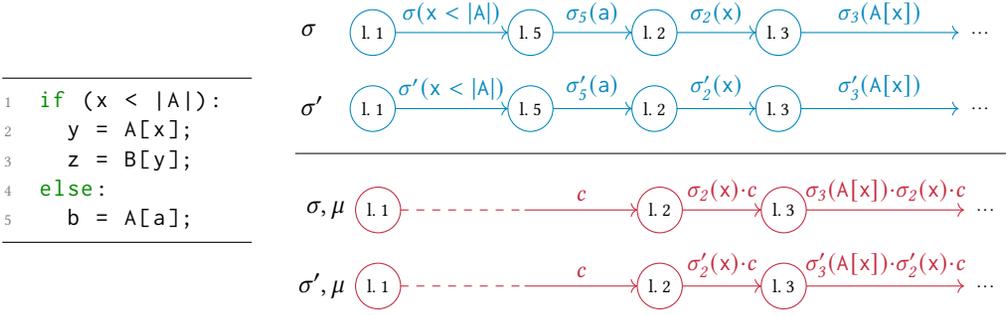
\begin{figure}[t]
		\small   
    \hspace{-5mm}
    \begin{subfigure}[b]{0.25\textwidth}
\begin{lstlisting}[style=Cstyle, basicstyle=\ttfamily\small]
if (x < |A|):
  y = A[x];
  z = B[y];
else:
  b = A[a];
\end{lstlisting}
\vspace{8mm}
    \label{subfig:spectre}
    \end{subfigure}%
    \hspace{5mm}
    \begin{subfigure}[b]{0.68\textwidth}
      \parbox{\textwidth}{%
      \centering
      \begin{tikzpicture}[initial text=, 
				->,
				node distance=1.65cm,
        state/.style = {circle, minimum size=6mm,
					inner sep=0pt, outer sep=0pt, font=\normalsize},
				sstate/.style = {state, draw=stepcol},
        hstate/.style = {state, draw=hcol}
				]

        \node[sstate] (s10) {\scriptsize l. 1};
        \node[state, left=0.1cm of s10] {$\sigma\phantom{,\mu}$};
				\node[sstate, right = 1.5cm of s10] (s11) {\scriptsize l. 5};
				\node[sstate, right of = s11] (s12) {\scriptsize l. 2};
				\node[sstate, right of = s12] (s13) {\scriptsize l. 3};
				\node[state, right = 2.1cm of s13] (s14) {\scriptsize \ldots};

        \node[sstate, below=0.4cm of s10] (s20) {\scriptsize l. 1};
        \node[state, left=0.0cm of s20] {$\sigma'\phantom{,\mu}$};
				\node[sstate, right = 1.5cm of s20] (s21) {\scriptsize l. 5};
				\node[sstate, right of = s21] (s22) {\scriptsize l. 2};
				\node[sstate, right of = s22] (s23) {\scriptsize l. 3};
				\node[state, right = 2.1cm of s23] (s24) {\scriptsize \ldots};

				\draw[stepcol] (s10) edge[above] node[sloped,pos=0.5]{\cstepstyle{\sigma(\texttt{x} < |\texttt{A}|)}} (s11)
        (s11) edge[above] node[sloped,pos=0.5]{\cstepstyle{\sigma_5(\texttt{a})}} (s12)
        (s12) edge[above] node[sloped,pos=0.5]{\cstepstyle{\sigma_2(\texttt{x})}} (s13)
        (s13) edge[above] node[sloped,pos=0.5]{\cstepstyle{\sigma_3(\texttt{A}[\texttt{x}])}} (s14);
        \draw[stepcol] (s20) edge[above] node[sloped,pos=0.5]{\cstepstyle{\sigma'(\texttt{x} < |\texttt{A}|)}} (s21)
        (s21) edge[above] node[sloped,pos=0.5]{\cstepstyle{\sigma'_5(\texttt{a})}} (s22)
        (s22) edge[above] node[sloped,pos=0.5]{\cstepstyle{\sigma'_2(\texttt{x})}} (s23)
        (s23) edge[above] node[sloped,pos=0.5]{\cstepstyle{\sigma'_3(\texttt{A}[\texttt{x}])}} (s24);
		\end{tikzpicture}
    }

    \vspace{.2cm}

    \rule{\textwidth}{.2pt}

    \vspace{.2cm}

    \parbox{\textwidth}{%
      \centering
			\begin{tikzpicture}[initial text=, 
				->,
				node distance=1.65cm,
        state/.style = {circle, minimum size=6mm,
					inner sep=0pt, outer sep=0pt, font=\normalsize},
				sstate/.style = {state, draw=stepcol},
        hstate/.style = {state, draw=hcol}
				]

        \node[hstate] (h20) {\scriptsize l. 1};
        \node[state, left=0.1cm of h20] {$\sigma, \mu$};

        \node[hstate, below=0.4cm of h20] (h10) {\scriptsize l. 1};
        \node[state, left=0.1cm of h10] {$\sigma', \mu$};

        \node[state, right of = h20] (h21) {};
        \node[hstate, right = 1.5cm of h21] (h22) {\scriptsize l. 2};
        \node[hstate, right of = h22] (h23) {\scriptsize l. 3};
        \node[state, right = 2.1cm of h23] (h24) {\scriptsize \ldots};

        \node[state, right of = h10] (h11) {};
        \node[hstate, right = 1.5cm of h11] (h12) {\scriptsize l. 2};
        \node[hstate, right of = h12] (h13) {\scriptsize l. 3};
        \node[state, right = 2.1cm of h13] (h14) {\scriptsize \ldots};

        \draw[hcol, -, dashed] (h10) -- (h11.east);
        \draw[hcol, -, dashed] (h20) -- (h21.east);

        \draw[hcol]
          (h11) edge[above] node[sloped,pos=0.5]{\hwstepstyle{c}} (h12)
          (h12) edge[above] node[sloped,pos=0.5]{\hwstepstyle{\sigma'_2(\texttt{x}){\cdot}c}} (h13)
          (h13) edge[above] node[sloped,pos=0.5]{\hwstepstyle{\sigma'_3(\texttt{A}[\texttt{x}]){\cdot}\sigma_2'(\texttt{x}){\cdot}c}} (h14);

        \draw[hcol]
          (h21) edge[above] node[sloped,pos=0.5]{\hwstepstyle{c}} (h22)
          (h22) edge[above] node[sloped,pos=0.5]{\hwstepstyle{\sigma_2(\texttt{x}){\cdot}c}} (h23)
          (h23) edge[above] node[sloped,pos=0.5]{\hwstepstyle{\sigma_3(\texttt{A}[\texttt{x}]){\cdot}\sigma_2(\texttt{x}){\cdot}c}} (h24);
		  \end{tikzpicture}
    }
    \label{subfig:traces}
    \end{subfigure}
		\caption{Program vulnerable to Spectre and potential contract (above) and hardware traces (below).
    The nodes indicate the line being executed next. The labels correspond to the leakage.
    $\sigma$ and $\sigma'$ are two initial architectural states, $\mu$ is the initial microarchitectural state, and $c$ is the initial cache.
    $\sigma_2(')$, $\sigma_3(')$, and $\sigma_5(')$ are the states after executing lines \texttt{2}, \texttt{3}, and \texttt{5}, respectively.
    We assume that the branch condition evaluates to true and that the hardware predicts the correct branch in state $\mu$.
    }
    \Description{An example of program vulnerable to Spectre}
    \label{fig:spectre}
  \end{figure}

  The program on the left in \Cref{fig:spectre} is a typical example demonstrating Spectre-style attacks.  
  The program first checks if \lstinline{x} is within the bounds of an array \lstinline|A|. If yes, it loads from \lstinline|A| at index \lstinline|x|. With branch misprediction, it might happen that the \texttt{if}-branch is executed even if \lstinline|x| is not within the bounds of \lstinline|A|. Then, a value from somewhere in the memory is loaded and stored in \lstinline|y|. This value, potentially a secret, is subsequently used as an address and thus leaked into the cache.

  \paragraph{Always-Mispredict Technique}
  The always-mispredict technique~\cite{spectector} (used in the CT-SPEC contract in~\cite{hwswcontracts}) abstracts away the microarchitectural branch predictor and instead always executes both branches.
  \Cref{fig:spectre} depicts exemplary traces generated by a contract using the always-mispredict technique (upper box, in blue) and a hardware semantics that speculates on branching decisions (lower box, in red).
  The four traces are obtained from Definition~\ref{def:hwswcontracts} when considering the program given on the left and two initial architectural states $\sigma$ and $\sigma'$ and initial microarchitectural state $\mu$.
  The leakage labels of the contract traces follow the constant-time principle, exposing the (correct) result of branching conditions and the addresses of memory accesses.
  The hardware traces expose the current cache.
  For this example, we assume that the branch condition evaluates to true in both $\sigma$ and $\sigma'$.
  Still, the contract also executes the wrong branch for a fixed number of steps to simulate a potential misprediction.
  The hardware traces, on the other hand, do not mispredict in microarchitectural state $\mu$.
  Importantly, the effects of the contract misprediction on the architectural
  states are reverted after executing $\scriptsize l. 5$ on the wrong branch, and both the contract and the hardware
  traces are synchronized again at line $\scriptsize l. 2$ in the same states $\sigma_2$ and $\sigma'_2$.
  From this point on, the leakage associated with each of the four executions depends on the evaluation of \lstinline|a|, \lstinline|x|, and \lstinline|A[x]| in the successive states.
  Since the states coincide on the contract and the hardware sides,
  if the contract traces agree on the value of \lstinline|x| and \lstinline|A[x]|,
  the hardware traces are also equal.
  Thus, Definition~\ref{def:hwswcontracts} would be satisfied.

  \paragraph{Correctness of the Always-Mispredict Technique}
  Intuitively, any hardware speculating solely on branching outcomes should indeed satisfy the contract, as the always-mispredict contract only exposes ``more'' information if the hardware does not mispredict.
  Proving contract satisfaction according to Definition~\ref{def:hwswcontracts} is challenging, though.
  We need to keep track of four different traces. Two traces of the same type always progress in lockstep, but operate on different inputs $\sigma, \sigma'$. The contract and hardware traces, on the other hand, may proceed at different speeds and produce different observation labels. Especially the former constitutes a challenge in formal proofs, as we cannot just perform a simple induction on the length of one of the contract traces.
  The setting calls for a simulation-style proof, but this requires finding suitable invariants across the four traces.
  In the following, we present a proof system that addresses these challenges by combining coinductive and inductive reasoning. We will come back to the always-mispredict technique in \Cref{sec:examples}, where we will give a formal semantics for a small ISA and prove the technique correct.

  \section{Proving Contract Satisfaction by Relative Bisimulation}
  \label{sec:relative-bisim}

  In this section, we present a new technique to formally prove contract satisfaction as defined in Definition~\ref{def:hwswcontracts}.
  We first emphasize that contract satisfaction is an instance of the more general problem of proving statements of the form ``trace equality implies trace equality''. We refer to this problem as \emph{relative trace equality}.
  We then show that relative trace equality is fully
  characterized by a 4-ary variant of bisimilarity which we
  call \emph{relative bisimilarity}.
  Building on relative bisimilarity, we develop a sound-and-complete deductive system to prove relative trace equality.
  This system offers simple reasoning principles for \emph{establishing} equality of traces between two states while \emph{exploiting} the assumption that two other states are trace-equal.

  \subsection{The Challenges of Trace-Based Reasoning}
  \label{subsec:mainidea}

  If we omit assumptions on initial states for a moment, proving contract satisfaction essentially boils down 
  to proving statements of the form $\ceval{P}(s_1) = \ceval{P}(s_2) \implies \hweval{P}(h_1) = \hweval{P}(h_2)$,
  where $s_1, s_2$ are two contract states and $h_1, h_2$ are two hardware states.
  Executing $P$ from the contract states $s_1, s_2$ and the hardware states $h_1, h_2$ leads to four executions:
  \begin{align*}
    s_1 \lstep{\mi{cobs}_1} s'_1 \lstep{\mi{cobs}'_1} s''_1 \lstep{\mi{cobs}''_1} \ldots
    \qquad
    h_1 \hwstep{\mi{hobs}_1} h'_1 \hwstep{\mi{hobs}'_1} h''_1 \hwstep{\mi{hobs}''_1} \ldots
    \\
    s_2 \lstep{\mi{cobs}_2} s'_2 \lstep{\mi{cobs}'_2} s''_2 \lstep{\mi{cobs}''_2} \ldots
    \qquad
    h_2 \hwstep{\mi{hobs}_2} h'_2 \hwstep{\mi{hobs}'_2} h''_2 \hwstep{\mi{hobs}''_2} \ldots
  \end{align*}%
  where $~\lstep{}~$ and $\hwstep{}$ are the small-step operational semantics defining the contract and the hardware model.
  Hardware labels $\mi{hobs}$ are used to model
  \emph{microarchitectural} information an attacker can observe during executions.
  Contract labels $\mi{cobs}$ are \emph{architectural} approximations of what is
  leaked to attackers.
  The four sequences of observations define the following \emph{traces}: \begin{align*}
    \ceval{P}(s_1) = \mi{cobs}_1~\mi{cobs}'_1~\mi{cobs}''_1\ldots
    \qquad
    \hweval{P}(h_1) = \mi{hobs}_1~\mi{hobs}'_1~\mi{hobs}''_1\ldots\\
    \ceval{P}(s_2) = \mi{cobs}_2~\mi{cobs}'_2~\mi{cobs}''_2\ldots
    \qquad
    \hweval{P}(h_2) = \mi{hobs}_2~\mi{hobs}'_2~\mi{hobs}''_2\ldots
  \end{align*}
  We want to prove that, assuming the two contract traces are equal, the hardware traces are also equal (i.e., the contract correctly over-approximates the leakage):
  \[
    \mi{cobs}_1~\mi{cobs}'_1~\mi{cobs}''_1\ldots = \mi{cobs}_2~\mi{cobs}'_2~\mi{cobs}''_2\ldots
    \implies
    \mi{hobs}_1~\mi{hobs}'_1~\mi{hobs}''_1\ldots = \mi{hobs}_2~\mi{hobs}'_2~\mi{hobs}''_2\ldots
  \]

  If contract traces were guaranteed to be finite and of equal length, we could reason by induction on the length of the two contract traces.
  Unfortunately, depending on the program $P$ and the states $s_1, s_2$,
  the executions might be of different lengths.
  They are not even guaranteed to both be finite as the same program $P$ might diverge from a state $s_1$ and terminate from another state $s_2$.
  If both traces are of infinite length, we cannot employ induction anyway as infinite traces are not inductive objects. 
  The exact same problems arise for hardware traces.

  \subsection{Relative Trace Equality Between Transition Systems}

  We propose a general solution to the problem of proving relative trace equality, i.e., statements of the form ``trace equality implies trace equality``.
  We do not commit to a specific notion of contract or hardware semantics.
  Instead, we treat both the contract and the hardware model as arbitrary \emph{deterministic} and \emph{non-terminating} state transition systems equipped with an infinite trace semantics.

  \begin{definition}[Transition Systems]
    A \emph{transition system} $\transsys$ is a tuple $(S, \Sigma, \trans, \leak)$, where $S$ is a set of states, $\Sigma$ is a set of observations,
    $\trans : S \to S$ is a transition function, and $\leak : S \to \Sigma$ is
    a \emph{leakage function}.
  \end{definition}

  \begin{definition}[Infinite Traces]
    Given a transition system $\transsys = (S, \Sigma, \trans, \leak)$,
    and a state $s \in S$, we
    denote by $\den{s}_\transsys \in \Sigma^\omega$ the (unique) \emph{infinite leakage trace} starting in $s$ \[
      \den{s}_\transsys \triangleq \leak(s)~\leak(\trans(s))~\leak(\trans^2(s))\ldots
    \]
  \end{definition}

  We note that assuming determinism is in line with previous work on hardware-software contracts~\cite{hwswcontracts, speculationatfault,Leave,mohr2024synthesizing,mohr2025synthesis}.
  We discuss this assumption and how to potentially lift it in \Cref{sec:discussion}.
  We also note that considering only non-terminating systems (and thus only infinite traces) does not result in any loss of generality, as terminating programs can
  be straightforwardly modeled in the infinite-trace setting by making terminating states self-loop forever.
  In fact, this approach offers more flexibility: it enables controlling whether non-termination is observable or not simply by choosing the labeling of terminal self-loops adequately.
  We omit the details here, which are quite standard, but we refer the curious readers to our Rocq development for a full formalization of the encoding(s) of termination in the infinite-trace setting (including proofs of correctness).

  We note that contract satisfaction (Definition~\ref{def:hwswcontracts}) can easily be stated in our transition system formalism.
  Indeed, let $P$ be a program expressed in some ISA, and let $\mathcal{T}^P_\sw = (S^P_\sw, \Sigma^P_\sw, \trans^P_\sw, \leak^P_\sw)$ and $\mathcal{T}^P_\hw = (S^P_\hw, \Sigma^P_\hw, \trans^P_\hw, \leak^P_\hw)$ be two transition systems representing a machine executing $P$ by following some contract semantics and hardware semantics, respectively (e.g., 
  if $s ~\lstep{\mi{cobs}}~ s'$, then $\trans^P_\sw(s) = s'$ and $\leak^P_\sw(s) = \mi{cobs}$, etc).
  Then, proving contract satisfaction means proving the following implication \[
    \forall P, \sigma, \sigma', \mu. \ \den{\sigma}_{\mathcal{T}^P_\sw} = \den{\sigma'}_{\mathcal{T}^P_\sw} \implies \den{(\sigma, \mu)}_{\mathcal{T}^P_\hw} = \den{(\sigma', \mu)}_{\mathcal{T}^P_\hw}
  \]

  \subsection{A State-Based Proof Technique: Relative Bisimulation}

  For the remainder of the section, let two transition systems $\sw= (S_\sw, \Sigma_\sw, \trans_\sw, \leak_\sw)$ (playing the role of a contract) and $\hw = (S_\hw, \Sigma_\hw, \trans_\hw, \leak_\hw)$ (playing the role of a hardware model) be given, and let $s_1, s_2 \in S_\sw$ and $h_1, h_2 \in S_\hw$ be four states.
  In the remainder of the paper, we often refer to $s_1$ and $s_2$ as the ``contract states'', and $h_1$ and $h_2$ as the ``hardware states'' for convenience.
  We want to obtain a technique to prove statements of the form \[
    \den{s_1}_\sw = \den{s_2}_\sw \implies \den{h_1}_\hw = \den{h_2}_\hw
  \]

  Ideally, to avoid the issues mentioned in \Cref{subsec:mainidea}, the technique should be state-based rather than trace-based.
  In the remainder of this subsection, we develop such a state-based proof technique as an extension of the standard bisimulation method for proving trace equality.
  
  \paragraph*{Fixpoints and (Co)Induction}

  Before introducing our proof technique, we quickly re-iterate over the necessary basics of fixpoints and their connection to proofs by (co)induction.
  Let $U$ be a set, and let $F : \mathcal{P}(U) \to \mathcal{P}(U)$ be a function from subsets of $U$ to subsets of $U$.
  A subset $X \subseteq U$ is called a \emph{fixpoint} of $F$ if $X = F(X)$. A famous result due to Tarski \cite{tarski} guarantees that if $F$ is \emph{monotone} (i.e., $\forall X, Y. \ X \subseteq Y \Rightarrow F(X) \subseteq F(Y)$), then it has a unique \emph{greatest fixpoint} (noted $\nu F$) and a unique \emph{least fixpoint} (noted $\mu F$).
  Further, least fixpoints come associated with a complete proof principle to show inclusions of the form $\mu F \subseteq X$ (induction). Dually, greatest fixpoints enjoy a complete proof principle for inclusions of the form $X \subseteq \nu F$ (coinduction).
  \begin{theorem}[Fixpoint Induction and Coinduction]
    \begin{align*}
      \pcomment{Induction} \mu F \subseteq X &\iff \exists I. \ I \subseteq X \wedge F(I) \subseteq I\\
      \pcomment{Coinduction} X \subseteq \nu F &\iff \exists I. \ X \subseteq I \wedge I \subseteq F(I)
    \end{align*}
  \end{theorem}
  
  Coinduction and induction can be seen as methods to prove inclusions by ``exhibiting an invariant". For example, coinduction can be informally spelled out as follows:
  to prove that elements of $X$ satisfy a property defined as a greatest fixpoint (i.e., $X \subseteq \nu F$), it suffices to find an ``invariant'' $I$ which holds initially (i.e., $X \subseteq I$), and which is preserved over one $F$-step (i.e., $I \subseteq F(I)$). $I$ is usually called a \emph{coinduction hypothesis}.

  \paragraph*{Bisimulation Proofs}

  Bisimulation is a standard state-based technique to prove that two states of a transition system generate the same traces (e.g., $\den{h_1}_\hw = \den{h_2}_\hw$).
  Intuitively, given a transition system $\transsys = (S, \Sigma, \trans, \leak)$, two states $s_1, s_2 \in S$ are said to be \emph{bisimilar} if they have the same immediate observations (i.e., $\leak(s_1) = \leak(s_2)$) and their successors are again bisimilar.
  More formally, the bisimilarity relation $\texttt{bisim} \subseteq S \times S$ is defined as a greatest fixpoint.

   \begin{definition}[Bisimilarity]
    $\texttt{bisim} \triangleq \nu\texttt{bisimF}$ where
    \begin{align*}
      \texttt{bisimF}(R) &\triangleq \{ \ (s_1, s_2) \mid \leak(s_1) = \leak(s_2) \ \wedge (\trans(s_1), \trans(s_2)) \in R \ \}
    \end{align*}
  \end{definition}

  An important result is that bisimilarity is equivalent to trace equality.

  \begin{theorem}[Bisimilarity is Trace Equality]
    \label{thm:bisim-sound}
    $(s_1, s_2) \in \texttt{bisim} \iff \den{s_1}_\transsys = \den{s_2}_\transsys$
  \end{theorem}

  Further, since \texttt{bisim} is defined as a greatest fixpoint, one can prove bisimilarity by coinduction.
  In particular, to prove that two states $s_1$ and $s_2$ are bisimilar, it suffices to find a relation $R$
  (usually called a \emph{bisimulation}) which contains $(s_1, s_2)$ and satisfies $R \subseteq \texttt{bisimF}(R)$.
  This process can be thought of as proving trace equality
  by exhibiting a suitable \emph{relational invariant}.

  \begin{theorem}[Proofs by Bisimulation]
    $ \den{s_1}_\transsys = \den{s_2}_\transsys \iff \exists R. (s_1, s_2) \in R \wedge R \subseteq \texttt{bisimF}(R)$
  \end{theorem}

  Remember that we aim to prove implications $\den{s_1}_\sw = \den{s_2}_\sw \implies \den{h_1}_\hw = \den{h_2}_\hw$.
  Omitting the assumption $\den{s_1}_\sw = \den{s_2}_\sw$, we could in principle prove $\den{h_1}_\hw = \den{h_2}_\hw$ by bisimulation.
  Unfortunately, in the context of contract satisfaction, this equality does not hold by itself: we \emph{need} to exploit the assumption $\den{s_1}_\sw = \den{s_2}_\sw$.
  Bisimulation alone is therefore not a suitable technique.

  \paragraph*{Relative Bisimilarity}
  To prove trace equality under a trace equality assumption,
  we generalize the standard binary notion of bisimilarity to a 4-ary \emph{relative bisimilarity} relation.
  Let $s_1, s_2 \in S_\sw$ and $h_1, h_2 \in S_\hw$ be four states.
  To prove that $\den{s_1}_\sw = \den{s_2}_\sw$ implies $\den{h_1}_\hw = \den{h_2}_\hw$ in a bisimulation-flavored way, the most straightforward approach would be to process the traces generated by the 4 states in lockstep, checking that \emph{local} equality of leakage on the left implies \emph{local} equality of leakage on the right.
  Following this intuition, we say that $s_1, s_2, h_1, h_2$ are \emph{lockstep relatively bisimilar} if either $\leak_\sw(s_1) \neq \leak_\sw(s_2)$, or if $\leak_\hw(h_1) = \leak_\hw(h_2)$ and $\trans_\sw(s_1),\trans_\sw(s_2), \trans_\hw(h_1), \trans_\hw(h_2)$ are again lockstep relatively bisimilar.
  Lockstep relative bisimilarity is again defined as a greatest fixpoint.

  \begin{definition}[Lockstep Relative Bisimilarity]
    $\texttt{rbisim}_\texttt{lockstep} \triangleq \nu\texttt{rbisimF}_\texttt{lockstep}$ where
    \begin{align*}
      \texttt{rbisimF}_\texttt{lockstep}(R) &\triangleq\\
      \pcomment{Leak} & \{ \ (s_1, s_2, h_1, h_2) \mid \leak_\sw(s_1) \neq \leak_\sw(s_2) \ \} \ \cup\\
      \pcomment{Step} & \{ \ (s_1, s_2, h_1, h_2) \mid \leak_\hw(h_1) = \leak_\hw(h_2) \\
      & \phantom{ \ (s_1, s_2, h_1, h_2) \mid} ~ \wedge (\trans_\sw(s_1), \trans_\sw(s_2), \trans_\hw(h_1), \trans_\hw(h_2)) \in R \ \}
    \end{align*}
  \end{definition}

  It is not difficult to see that
  lockstep relative bisimilarity is a sound approximation of relative trace equality: if four states are lockstep relatively bisimilar, they also satisfy relative trace equality.

  \jana{Change to lemma? As this is not one of our main results.}
  \begin{theorem}
    $(s_1, s_2, h_1, h_2) \in \texttt{rbisim}_\texttt{lockstep} \implies (\den{s_1}_\sw = \den{s_2}_\sw \implies \den{h_1}_\hw = \den{h_2}_\hw)$
  \end{theorem}

  Unfortunately the reverse implication does not hold:
  $\texttt{rbisim}_\texttt{lockstep}$
  is \emph{not} a complete characterization of relative trace equality.
  As a counterexample, consider the following scenario where $\mi{obsA}$, $\mi{obsB}$, $\mi{obsC}$, $\mi{obsD}$, $\mi{obsE}$ are 5 \emph{distinct} observations:\begin{align*}
    s_1 \xrightarrow{\mi{obsA}} s'_1 \xrightarrow{\mi{obsB}} \ldots
    \qquad
    h_1 \xrightarrow{\mi{obsD}} \ldots
    \\
    s_2 \xrightarrow{\mi{obsA}} s'_2 \xrightarrow{\mi{obsC}} \ldots
    \qquad
    h_2 \xrightarrow{\mi{obsE}} \ldots
  \end{align*}%

  Since $\mi{obsB} \ne \mi{obsC}$, the trace of $s_1$ is different from the trace of $s_2$ and relative trace equality trivially holds for $s_1, s_2, h_1, h_2$. However, the four states are \emph{not} relatively bisimilar according to $\texttt{bisim}_\texttt{lockstep}$ because neither $\leak(s_1) \ne \leak(s_2)$ nor $\leak(h_1) = \leak(h_2)$.
  Here, the problem comes from the lockstep nature of $\texttt{bisim}_\texttt{lockstep}$.
  Concretely, we would like to be allowed to "skip" steps on the left-hand side, independently of the right-hand side.
  In the previous example, this would have allowed us to fast-forward to the ``contract states'' $s'_1, s'_2$, and then exploit the inequality $\mi{obsB} \neq \mi{obsC}$ without having to prove anything about the ``hardware states'' in between.
  Such examples where we need non-lockstep
  arguments abound in practice, particularly in the context of contract satisfaction.
  To support non-lockstep arguments, a naive idea would be to slightly update the definition of $\texttt{rbisim}_\texttt{lockstep}$ to decouple ``contract steps'' and ``hardware steps''.
  For example, we could relax the definition of $\texttt{rbisim}_\texttt{lockstep}$ as follows:\begin{align*}
    \texttt{rbisim}_\texttt{relaxed} &\triangleq \nu\texttt{rbisimF}_\texttt{relaxed}\\
    \texttt{rbisimF}_\texttt{relaxed}(R) &\triangleq\\
    \pcomment{C-Leak} & \{ \ (s_1, s_2, h_1, h_2) \mid \leak_\sw(s_1) \neq \leak_\sw(s_2) \ \} \ \cup\\
    \pcomment{C-Step} & \{ \ (s_1, s_2, h_1, h_2) \mid (\trans_\sw(s_1), \trans_\sw(s_2), h_1, h_2) \in R \ \} \ \cup\\
    \pcomment{H-Step} & \{ \ (s_1, s_2, h_1, h_2) \mid \leak_\hw(h_1) = \leak_\hw(h_2) \wedge (s_1, s_2, \trans_\hw(h_1), \trans_\hw(h_2)) \in R \ \}
  \end{align*}

  The addition of the disjunct $\comment{C-Step}$ lets us skip contract execution steps, as we wanted. In particular, the following implication now holds by definition \[
    (\trans_\sw(s_1), \trans_\sw(s_2), h_1, h_2) \in \texttt{rbisim}_\texttt{relaxed} \implies (s_1, s_2, h_1, h_2) \in \texttt{rbisim}_\texttt{relaxed}
  \]

  Unfortunately, with this naive modification, $\texttt{rbisim}_\texttt{relaxed}$ is no longer sound for relative trace equality.
  The problem is that
  considering the greatest fixpoint of $\texttt{rbisimF}_\texttt{relaxed}$ gives an overly permissive notion of relative bisimilarity: it enables proving that a relation $R$ is a (relaxed) relative bisimulation via \comment{C-Step} without ever proving anything useful about relative trace equality (i.e., without providing evidences that $\den{s_1}_\sw \neq \den{s_2}_\sw$, nor evidences that $\den{h_1}_\hw = \den{h_2}_\hw$).
  In fact, any quadruple of states vacuously satisfies $\texttt{rbisim}_\texttt{relaxed}$.
  Indeed, let $\top = S_\sw \times S_\sw \times S_\hw \times S_\hw$ be the set of all quadruple of states. Using \comment{C-Step}, we can trivially prove $\top \subseteq \texttt{bisimF}_\texttt{relaxed}(\top)$. By coinduction, it follows that $\top \subseteq \texttt{rbisim}_\texttt{relaxed}$ (i.e., all quadruples of states are relatively bisimilar according to the relaxed definition).
  To achieve soundness while preserving support for non-lockstep reasoning,
  the solution is to employ an alternation of greatest and least fixpoint, as demonstrated by the following definition.

  \begin{definition}[Relative Bisimilarity]
    \label{def:rbisim}
    $\texttt{rbisim} \triangleq \nu R_1. \, \texttt{rbisimF}(R_1)$ where
    \begin{align*}
      \texttt{rbisimF}(R_1) &\triangleq \mu R_2.\\
      \pcomment{C-Leak} & \{ \ (s_1, s_2, h_1, h_2) \mid \leak_\sw(s_1) \neq \leak_\sw(s_2) \ \} \ \cup\\
      \pcomment{C-Step} & \{ \ (s_1, s_2, h_1, h_2) \mid (\trans_\sw(s_1), \trans_\sw(s_2), h_1, h_2) \in R_2 \ \} \ \cup\\
      \pcomment{H-Step} & \{ \ (s_1, s_2, h_1, h_2) \mid \leak_\hw(h_1) = \leak_\hw(h_2) \wedge (s_1, s_2, \trans_\hw(h_1), \trans_\hw(h_2)) \in R_1 \ \}
    \end{align*}
  \end{definition}

  This definition closely resembles the definition of $\texttt{rbisim}_\texttt{relaxed}$, except that contract steps are now treated \emph{inductively} instead of \emph{coinductively}.
  The resulting notion of (non-lockstep) relative bisimilarity turns out to exactly coincide with relative trace equality.

  \begin{theorem}[Relative Bisimilarity is Relative Trace Equality]
    \label{lem:sound}
    \[
      (s_1, s_2, h_1, h_2) \in \texttt{rbisim} \iff (\den{s_1}_\sw = \den{s_2}_\sw \implies \den{h_1}_\hw = \den{h_2}_\hw)
    \]
  \end{theorem}
  \begin{proof}
    (\textbf{Soundness} ${\implies}$).
    Since bisimilarity implies trace equality (Theorem~\ref{thm:bisim-sound}), it suffices to prove $(s_1, s_2, h_1, h_2) \in \texttt{rbisim} \implies (\den{s_1}_\sw = \den{s_2}_\sw \implies (h_1, h_2) \in \texttt{bisim})$.
    The proof proceeds by coinduction on $\texttt{bisim}$, and by induction on the
    inductive part of $\texttt{rbisim}$.
    (\textbf{Completeness} ${\impliedby}$).
    We proceed by case distinction and first suppose $\den{s_1}_\sw = \den{s_2}_\sw$.
    By assumption it follows that $\den{h_1}_\hw = \den{h_2}_\hw$.
    Proving $(s_1, s_2, h_1, h_2) \in \texttt{rbisim}$ now degenerates
    into a standard bisimulation proof between $h_1$ and $h_2$ using
    only \comment{H-Step} and ignoring $s_1$ and $s_2$.
    For the other case, let $\den{s_1}_\sw \neq \den{s_2}_\sw$, so there must
    exist some $n \in \mathbb{N}$ such that $\trans^n_\sw(s_1) \neq \trans^n_\sw(s_2)$.
    Under this assumption, we can prove $(s_1, s_2, h_1, h_2) \in \texttt{rbisim}$
    by induction on $n$ (using \comment{Leak} in the base case
    and \comment{C-Step} in the inductive case).
  \end{proof}

  As a consequence of Theorem~\ref{lem:sound}, we obtain a sound and complete invariant-based technique to prove relative trace equality.
  Indeed, since $\texttt{rbisim}$ is defined as a greatest fixpoint, we can prove relative bisimilarity (and hence relative trace equality) by coinduction: it suffices to find a relation $R$ containing the four initial states and satisfying $R \subseteq \texttt{rbisimF}(R)$.

  \begin{theorem}[Proofs by Relative Bisimulation]
    \label{thm:rbisim-proofs}
    \[
      (\den{s_1}_\sw = \den{s_2}_\sw \implies \den{h_1}_\hw = \den{h_2}_\hw) \Leftrightarrow \exists R. \, (s_1, s_2, h_1, h_2) \in R \wedge R \subseteq \texttt{rbisimF}(R)
    \]
  \end{theorem}

  \paragraph*{Weak Relative Bisimilarity}

  Our final definition of relative bisimilarity (\texttt{rbisim}) enables a form of
  non-lockstep reasoning where ``contract steps'' can be decoupled from ``hardware steps''.
  However, the two contract states are progressing in locksteps, and so are the two hardware states.
  This is required in our setting, because we work with a strong notion of trace equality where equality between observations must hold at each step.
  However, it is sometimes useful to consider a relaxed setting where some of the computation steps emit so-called \emph{tau events} (noted $\tau$), indicating that no observation is produced.
  In this case, traces are only required 
  to be equal up to removal of $\tau$'s (usually called weak trace equality), and non-lockstep reasoning between the two contract states (or the two hardware states) becomes a particularly useful feature to skip silent steps and realign states as needed.
  It is well-know that bisimilarity can be weakened to
  support weak trace equality.
  Similarly, we can adapt relative bisimilarity to 
  support \emph{weak relative trace equality} (i.e., weak trace equality implies weak trace equality).
  Concretely, we can define a relation $\texttt{weak\_rbisim}$ which enables skipping $\tau$-labeled steps from any of the four states without changing the other three.
  For the purpose of clarity, and since recent works on hardware-software contracts did not make use of $\tau$ labels, the rest of the paper also sticks to the strict setting without $\tau$'s.
  We refer the curious readers to our Rocq development for the definition and the soundness proof of $\texttt{weak\_rbisim}$.
  We leave the completeness proof as future work.
  
  \subsection{A Deductive System for Relative Bisimilarity}
  \label{sec:proofsystem}

  As stated in Theorem~\ref{thm:rbisim-proofs}, relative bisimilarity provides a
  sound and complete state-based technique to prove relative trace equality.
  Still, finding a suitable relational invariant (call it $R$) and checking that it is indeed a relative bisimulation (i.e., $R \subseteq \texttt{rbisimF}(R)$) can be tedious.
  Instead, we propose a simple deductive system to prove relative bisimilarity.
  Our deductive system is based on \emph{parameterized coinduction} \cite{paco}, and enables a modular and incremental proof style that progressively builds a relative bisimulation rather than having to provide it up-front.

  \paragraph*{Parameterized Coinduction}

  Parameterized coinduction (Paco) was introduced by Hur et al. \cite{paco} as a technique to simplify proofs by coinduction in interactive proof assistants.
  The key idea is to replace predicates defined as greatest fixpoints $\nu F$ by \emph{parameterized predicates} $G_F(I) \triangleq \nu X. \ F(X \cup I)$ where $I$ can be thought of as a coinduction hypothesis ``in construction''.
  By definition it is clear that $G_F(\emptyset) = \nu F$, which means
  that proving $X \subseteq \nu F$ is equivalent to proving $X \subseteq G_F(\emptyset)$.
  The benefit of replacing
  $\nu F$ with $G_F(\emptyset)$ is that it supports useful rules to construct an invariant incrementally.
  In particular, the following three rules can be derived from the definition of $G_F(I)$ \cite{paco}.
  \begin{mathpar}
    \inferrule[Paco-Init]{
      X \subseteq G_F(\emptyset)
    }{
      X \subseteq \nu F
    }
    \qquad
    \inferrule[Paco-Step]{
      X \subseteq F(I \cup G_F(I))
    }{
      X \subseteq G_F(I)
    }\qquad
    \inferrule[Paco-Accumulate]{
      X \subseteq G_F(I \cup X)
    }{
      X \subseteq G_F(I)
    }
  \end{mathpar}

  When proving a statement of the form $X \subseteq \nu F$, instead of directly using standard coinduction (which requires guessing $I \supseteq X$ such that $I \subseteq F(I)$), we can use \textsc{Paco-Init} to prove $X \subseteq G_F(\emptyset)$ instead.
  From there, we use the rules 
  \textsc{Paco-Step} and \textsc{Paco-Accumulate} to progressively enlarge the parameter $I$ until it is big enough to directly prove $X \subseteq F(I)$ using \textsc{Paco-Step}.
  In the context of a proof by bisimulation (i.e., when $F = \texttt{bisimF}$), this intuitively corresponds to starting with an empty bisimulation, and following transitions (using \textsc{Paco-Step}) while accumulating pairs of states (using \textsc{Paco-Accumulate}) until we obtain a valid bisimulation relation.
  Importantly, not all states need to be accumulated.
  In many cases, the parameter $I$ does not need to be a full bisimulation relation to conclude a proof, but rather a small fragment of a bisimulation.
  In extreme cases, it even suffices to accumulate a single pair of states to justify the existence of a bisimulation.

  \paragraph*{Proof Quintuples}

  Remember that relative bisimilarity is defined as a greatest fixpoint $\texttt{rbisim} = \nu R. \, \texttt{rbisimF}(R)$.
  This means we can already use the rules of Paco to prove relative trace equality (just take $F = \texttt{rbisimF}$).
  To simplify notations, we avoid referring to $G_\texttt{bisimF}$ directly, and instead introduce the following
  \emph{proof quintuples} to denote the possible states of a proof.
  \begin{definition}[Proof Quintuples]
    \label{def:quintuples}
    Let $F \triangleq \texttt{rbisimF}$. We define
    \begin{align*}
      \gquadruple{R}{s_1}{s_2}{h_1}{h_2} \triangleq (s_1, s_2, h_1, h_2) \in G_F(R)\qquad
      \quadruple{R}{s_1}{s_2}{h_1}{h_2} \triangleq (s_1, s_2, h_1, h_2) \in (R \cup G_F(R))
    \end{align*}
  \end{definition}
  
  The quintuple with a solid box corresponds to a state of a proof ``before a Paco step'', and a dashed box corresponds to a state ``after a Paco step''.
  After using \textsc{Paco-Step}, the current coinduction hypothesis $R$ can be used to conclude the proof. For this reason,
  we often refer to $\fbox{$R$}$ as a \emph{guarded hypothesis} and $\dbox{$R$}$ as an \emph{unguarded hypothesis}.
  This notation (and naming convention) is borrowed from recent works on parameterized coinduction \cite{hyco, almostfair, sff}.

  Directly using the rules of Paco is still not very helpful though: the rules would be formulated in terms of \texttt{rbisimF},
  which hides an inner least fixpoint that users still have to
  handle manually.
  Instead, we specialize the rules of Paco to obtain reasoning principles that are specific to relative bisimilarity.
  In particular, we formulate our reasoning rules in terms of atomic contract steps and hardware steps.

  \begin{theorem}[Core Reasoning Rules of Proof Quintuples]
    \label{def:core-rules}
    Quintuples admit the following rules:
    \begin{mathpar}
      \inferrule[C-Leak]{\leak_\sw(s_1) \neq \leak_\sw(s_2)}{{\gquadruple{R}{s_1}{s_2}{h_1}{h_2}}}
      \qquad
      \inferrule[C-Step]{
        {\gquadruple{R}{\trans_\sw(s_1)}{\trans_\sw(s_2)}{h_1}{h_2}}
      }{
        {\gquadruple{R}{s_1}{s_2}{h_1}{h_2}}
      }
      \qquad
      \inferrule[H-Step]{
        \leak_\hw(h_1) = \leak_\hw(h_2)\\\\
        \quadruple{R}{s_1}{s_2}{\trans_\hw(h_1)}{\trans_\hw(h_2)}
      }{
        {\gquadruple{R}{s_1}{s_2}{h_1}{h_2}}
      }
    \end{mathpar}
    \begin{mathpar}
      \inferrule[Invariant]{(s_1, s_2, h_1, h_2) \in R' \\\\ \forall (s'_1, s'_2, h'_1, h'_2) \in R'. \ {\gquadruple{R \cup R'}{s'_1}{s'_2}{h'_1}{h'_2}}}{{\gquadruple{R}{s_1}{s_2}{h_1}{h_2}}} \qquad
      \inferrule[Cycle]{(s_1, s_2, h_1, h_2) \in R}{{\quadruple{R}{s_1}{s_2}{h_1}{h_2}}}\qquad
      \inferrule[Guard]{{\gquadruple{R}{s_1}{s_2}{h_1}{h_2}}}{{\quadruple{R}{s_1}{s_2}{h_1}{h_2}}}
    \end{mathpar}
  \end{theorem}

  \paragraph{Step Rules}

  The three rules \textsc{C-Leak}, \textsc{C-Step} and \textsc{H-Step} are derived from \textsc{Paco-Step} and correspond to the three disjuncts with the same names in the definition of relative bisimilarity (Definition~\ref{def:rbisim}).
  On the contract side, the rule \textsc{C-Leak} enables concluding a proof whenever the two contract states have different observations ($\leak_\sw(s_1) \neq \leak_\sw(s_2)$), and the
  the rule \textsc{C-Step} just skips one contract execution step.
  On the hardware side, the rule \textsc{H-Step} requires proving that $\leak_\hw(h_1) = \leak_\hw(h_2)$, and in exchange it replaces 
  the current hardware states by their successors.
  
  Note that \textsc{H-Step} releases the guard of $R$,
  while \textsc{C-Step} does not.
  This asymmetry stems from the fact that hardware steps are treated coinductively, while contract steps are treated inductively within the definition of \texttt{rbisim}.
  Since we want to prove $\den{s_1}_\sw = \den{s_2}_\sw \implies \den{h_1}_\hw = \den{h_2}_\hw$, this difference makes sense:
  taking a hardware step makes progress towards proving $\den{h_1}_\hw = \den{h_2}_\hw$, while taking a contract step just ``unrolls'' the assumption $\den{s_1}_\sw = \den{s_2}_\sw$.

  \paragraph*{Cycle and Guard Rules}
  Whenever the guard around the hypothesis is released (i.e., after using the rule \textsc{H-Step}), we can use the rule \textsc{Cycle} to immediately conclude a proof if the four current states are already contained in $R$.
  If not, we restore the guard
  using the rule \textsc{Guard}, and we need to keep going (either by taking more steps, or by extending the hypothesis using \textsc{Invariant}).

  \paragraph*{Invariant Rule}

  The last rule is the \textsc{Invariant} rule, which is derived from \textsc{Paco-Accumulate}.
  It is used to enlarge the current hypothesis by ``guessing'' a relational invariant.
  More precisely, if the current states $(s_1, s_2, h_1, h_2)$ satisfy a relation $R'$,
  we can replace the current hypothesis $R$ with $R \cup R'$.
  In exchange, we then need to prove a quintuple for \emph{all} states in $R'$, not just the current states.
  We will be able to exploit this new hypothesis via the \textsc{Cycle} rule if we re-encounter states satisfying $R'$ later on.

  \paragraph*{Enhanced \textsc{C-Step} Rule}

  We note that \textsc{C-Step} and \textsc{C-Leak} can be combined into a convenient
  third rule to simultaneously skip a contract step and assume equality of (contract) leakage: \begin{mathpar}
    \inferrule[C-Step']{\leak_\sw(s_1) = \leak_\sw(s_2) \implies \gquadruple{R}{\trans_\sw(s_1)}{\trans_\sw(s_2)}{h_1}{h_2}}{{\gquadruple{R}{s_1}{s_2}{h_1}{h_2}}}
  \end{mathpar}
  In practice, we use this rule instead of \textsc{C-Step} because we typically need to exploit the assumption $\leak_\sw(s_1) = \leak_\sw(s_2)$ right after a contract step.

  \paragraph*{Soundness and Completeness of Quintuples}

  By definition, quintuples with an empty hypothesis are equivalent to relative bisimilarity, and hence equivalent to relative trace equality.

  \begin{theorem}
    \label{thm:quintuples-sound-complete}
    $\gquadruple{\emptyset}{s_1}{s_2}{h_1}{h_2} \iff (\den{s_1}_\sw = \den{s_2}_\sw \implies \den{h_1}_\hw = \den{h_2}_\hw)$
  \end{theorem}

  A result that is a little less obvious is that the six proof rules of Theorem~\ref{def:core-rules} constitute a complete deductive system for relative trace equality.

  \begin{theorem}[Completeness]
    Suppose that $(\den{s_1}_\sw = \den{s_2}_\sw \implies \den{h_1}_\hw = \den{h_2}_\hw)$, then we can prove $
      \fbox{$\emptyset$} \vdash \begin{bmatrix}
      s_1 & h_1\\
      s_2 & h_2
      \end{bmatrix}$
    using only the rules \textsc{C-Leak}, \textsc{C-Step}, \textsc{H-Step}, \textsc{Invariant}, \textsc{Cycle}, and \textsc{Guard}.
  \end{theorem}
  \begin{proof}
    Suppose $\den{s_1}_\sw = \den{s_2}_\sw \implies \den{h_1}_\hw = \den{h_2}_\hw$.
    By Theorem~\ref{thm:rbisim-proofs}, it follows that $(s_1, s_2, h_1, h_2) \in \texttt{rbisim}$, and we can use $\texttt{rbisim}$ itself as an invariant via the $\textsc{Invariant}$ rule. The rest of the proof proceeds by induction on the inductive part of \texttt{rbisim}. The three cases ($\comment{C-Leak}$, $\comment{C-Step}$, and $\comment{H-Step}$) are covered using the three rules with the same names.
    The case of $\comment{C-Step}$ is concluded by induction. The case of $\comment{H-Step}$ is concluded by \textsc{Cycle}.
  \end{proof}

  \section{Case Study}
  \label{sec:examples}

  We used the proof system presented in \Cref{sec:proofsystem} to formally prove two of the most important ideas underlying the CT-SPEC contract, which was developed for a hardware model that features branch speculation and out-of-order execution~\cite{hwswcontracts}.
  The first idea is to abstract away the branch predictor by employing a contract which
  always starts by mispredicting (as described in \Cref{sec:HW-SW-contracts}).
  The second idea is to abstract away the out-of-order execution of instructions by employing a contract which executes instructions sequentially.
  Our formalization of both abstractions use a simpler ISA and hardware semantics as opposed to~\cite{hwswcontracts}, but they contain the key features that make CT-SPEC such a useful abstraction.
  We present the first example, the always-mispredict abstraction, in detail and give a short overview of the second example. A more detailed description of the out-of-order example can be found in 
  \ifthenelse{\boolean{fullversion}}
  {the appendix in \Cref{app:OOOexample}}
  {the full version of the paper~\cite{fullVersion}}.

  \subsection{A Small ISA}

  We formulate a small toy ISA that consists of three instructions that manipulate two registers, and can load from memory. We omit stores for this example, as they behave just as loads from a leakage point of view. For similar reasons, we ignore expressions other than addition.

  \begin{center}
	\begin{tabular}{llcl}
	\textit{Registers} 	&  $x$		& $\in$ & $\set{r_1, r_2}$ \\
	\textit{Constants} 		&  $k$ 		& $\in$ & $\Int$  \\
	\textit{Locations} 		&  $l$ 		& $\in$ & $\Nat$  \\
	\textit{Instructions} 	&  $i$ 		& $:=$ & $\pload{x_1}{x_2} \mid \tadd{x_1}{x_2}{k} \mid \pjz{x}{l}$ \\ %
	\textit{Programs}		&  $p$		& $:=$ & $i \mid \pseq{p_1}{p_2} $%
	\end{tabular}
\end{center}

  An architectural state $\sigma = (m, a, \pc)$ consists a memory $m : \Nat \to \Int$, a register assignment $a : \set{r_1, r_2} \to \Int$, and a program counter $\pc : \Nat$.
  Technically, the program $P$ is also part of the architectural state, but we keep $P$ a global variable in the semantics.
  The hardware state $s = (\sigma, c)$ adds to the architectural state a cache $c$, storing the recently accessed addresses.
  
  The vanilla architectural and hardware semantics implement the expected behavior. 
  To simplify the presentation of subsequent examples, the semantics are given as a base relation describing the effect on the state when executing a single instruction. The full semantics are given in \ifthenelse{\boolean{fullversion}}
  {the appendix in \Cref{app:fullSemantics}}
  {the full version~\cite{fullVersion}}. As an example, the two rules for loads are the following:
  \begin{mathpar}
    \inferrule[LoadArch]
    { v = m(a(x_2)) }
    {
      \tup{m, a, \pc, \pload{x_1}{x_2}} \archsingleStep{} \tup{m, a[x_1 \mapsto v], \pc + 1}
    }

    \inferrule[LoadHW]
      {
        (\sigma, \pload{x_1}{x_2}) \archsingleStep{} \sigma'
      }
      {
        \tup{\sigma, c, \pload{x_1}{x_2}} \hwsingleStep{c} \tup{\sigma', a(x_2) \cons c}
      }
  \end{mathpar}
  Apart from state transformations, the hardware semantics also models a cache-based attacker by adding observation labels to the semantics. In this case, the attacker can observe the contents of the cache at any time, so we expose it in every step.
  For readability, we omitted from the semantics the fact that values, when used as addresses in the \textsc{Load} rules, are truncated to 0 if they are negative.

  \subsection{Speculation and Always-Mispredict Contract}
  \label{subsec:always-mispredict-example}

    \begin{figure}
  \small
  \begin{mathpar}
    \inferrule[Rollback]
    {
      s' = (\sigma_b, s.c)
    }
    {
      \tup{s, 0, (\mathsf{false}, \sigma_b)} \hwstep{s.c} \tup{s', \infty, \bot}
    }

    \inferrule[Commit]
    { }
    {
      \tup{s, 0, (\mathsf{true}, \sigma_b)} \hwstep{s.c} \tup{s, \infty, \bot}
    }
    \vspace{-1pt}

    \inferrule[Step]
    {
      \mathsf{cp} \neq \bot \lor P(s.\pc) \neq \pjz{\cdot}{\cdot} \\
      (s, P(s.\pc)) \hwsingleStep{o} s'
    }
    {
      \tup{s, \omega + 1, \mathsf{cp}} \hwstep{o} \tup{s', \omega, \mathsf{cp}}
    }
    \vspace{-4pt}

    \inferrule[Branch Next]
    {
      P(s.\pc) = \pjz{x}{l} \\
      \Phi (s.\pc) = \mi{next}\\
      (s, \pjz{x}{l}) \hwsingleStep{o} (\sigma', \_) \\
      \mathsf{b} = \sigma'.\pc \overset{?}{=} \, s.\pc + 1
    }
    {
      \tup{s, \infty, \bot} \hwstep{s.c} \tup{s[pc \mapsto s.\pc + 1], w, (\mathsf{b}, \sigma')}
    }
    \vspace{-4pt}

    \inferrule[Branch jump]
    {
      P(s.\pc) = \pjz{x}{l} \\
      \Phi (s.\pc) = \mi{jump}\\
      (s, \pjz{x}{l}) \hwsingleStep{o} (\sigma', \_) \\
      \mathsf{b} = \sigma'.\pc \overset{?}{=} \, l
    }
    {
      \tup{s, \infty, \bot} \hwstep{s.c} \tup{s[pc \mapsto l], w, (\mathsf{b}, \sigma')}
    }

    \end{mathpar}
  \caption{Hardware semantics with a branch predictor $\Phi$.}
  \Description{Hardware semantics with a branch predictor}
  \label{fig:mispredict-hard}
  \end{figure}

  \paragraph{Speculating Hardware Semantics}
  We now come back to our example from \Cref{subsec:sec2-always-mispredict} and formalize a speculating hardware semantics and an always-mispredict contract.
  We extend the vanilla hardware semantics with speculative execution in \Cref{fig:mispredict-hard}.
  The microarchitecture now contains a branch predictor $\Phi: \Nat \to \set{\mi{jump}, \mi{next}}$, which predicts the outcome of a branching decision based on the current program counter.
  This function is used as a global variable as it does not change during execution.
  It could easily be extended with a local state to model training of the predictor.
  We use $s.\pc, s.c$, and $s.a$ to access the components of the hardware state. We now also assume that a program $P$ is being executed, compared to before, where we only executed a single instruction. We write $P(s.\pc)$ to access the instruction that the program counter points to.

  The hardware now keeps track of a speculation window $\omega \in \Nat \cup \set{\infty}$ and a checkpoint pair $\mathsf{cp}$.
  When not in speculation, $\omega = \infty$ and $\mathsf{cp} = \bot$.
  Upon reaching a branching instruction (rules \textsc{Branch Next} and \textsc{Branch Jump}), we initialize the speculation window with a constant $w$. The checkpoint serves as an oracle that will be queried after $w$ steps to check if the prediction was correct (Boolean tag \textsf{b}) and where to rollback to (state $\sigma_b$) if the prediction was wrong. This is implemented in rules \textsc{Rollback} and \textsc{Commit} and models the fact that the hardware only learns later if the branching decision was correct.
  Our semantics does not allow nested speculation (this is further discussed in \Cref{sec:discussion}).

  \paragraph{Always-Mispredict Contract}
  We adapt the always-mispredict technique described in \Cref{subsec:sec2-always-mispredict} to our example with the contract given in \Cref{fig:mispredict-contr}. 
  Whenever the contract executes a load, the address of the load is exposed as an observation label.
  Similar to the speculating hardware semantics, the contract uses a speculation window $\omega$ and a rollback state $\sigma_b$. 
  Differently, the contract gets immediate access to the correct branch.
  This enables it to always first execute the wrong branch before returning to the correct one.
  Following the CT-SPEC contract, we also expose what the correct branching decision would be, which is an adoption of the constant-time model.

  \begin{figure}
  \small
  \begin{mathpar}

    \inferrule[Step]
    {
      \cstepstyle{o} =
        {
          \begin{cases}
            \cstepstyle{\sigma.a(x_2)} & \text{if } P(\sigma.\pc) = \pload{x_1}{x_2}\\
            \cstepstyle{\sigma.a(x) \,\overset{?}{=}\, 0} & \text{if } P(\sigma.\pc) = \pjz{x}{l} \\
            \cstepstyle{\bot} & \text{otherwise}
          \end{cases}
        } \\
      \sigma_b \neq \bot \lor P(\sigma.\pc) \neq \pjz{\cdot}{\cdot}\\
      (\sigma, P(\sigma.\pc)) \archsingleStep{o} \sigma' 
    }
    {
      \tup{\sigma, \omega + 1, \sigma_b} \amstep{o} \tup{\sigma', \omega, \sigma_b}
    }
    \vspace{-3pt}

    \inferrule[Rollback]
    { }
    {
      \tup{\sigma, 0, \sigma_b} \amstep{} \tup{\sigma_b, \infty, \bot}
    }

    \inferrule[Branch]
    {
      P(\sigma.\pc) = \pjz{x}{l} \\
      \hspace{-3pt}(\sigma, \pjz{x}{l}) \!\archsingleStep{o}\! \sigma' \\
      \hspace{-3pt}\pc_\bot = (a(x) \overset{?}{=} 0) ~?~ (\pc + 1) : l
    }
    {
      \tup{\sigma, \infty, \bot} \amstep{\sigma.a(x) \,\overset{?}{=}\, 0} \tup{\sigma[pc \mapsto \pc_\bot], w, \sigma'}
    }
    \end{mathpar}
  \caption{Always-mispredict contract.}
  \Description{Always-mispredict contract.}
  \label{fig:mispredict-contr}
  \end{figure}

  We note that even though all our semantics are presented in a small-step relational style for clarity, they can equivalently be stated in a functional style
  by defining separately two functions $\trans$ and $\leak$:
  for a small-step labeled relation $s \xrightarrow{o} s'$, 
  we have $s \xrightarrow{o} s' \iff s' = \trans(s) \wedge o = \leak(s)$.

  \subsection{Proof of the Always-Mispredict Contract}
  \label{example-am}

  We use our proof system to formally prove that the always-mispredict contract is satisfied by our simplified model of a speculating CPU.
  More precisely, for any program $P$, for any initial memories and register assignments $m_1, m_2, a_1, a_2$, for any initial cache $c$, and for any initial program counter $\pc$, we prove the following statement:
  \begin{align*}
    \ceval{P}{((m_1, a_1, \pc), \infty, \bot)} = \ceval{((m_2, a_2, \pc), \infty, \bot)} \implies\\
    \hweval{P}{((m_1, a_1, \pc, c), \infty, \bot)}{} = \hweval{P}{((m_2, a_2, \pc, c), \infty, \bot)}
  \end{align*}
  By the soundness of our proof system (Theorem \ref{thm:quintuples-sound-complete}), it suffices to instead prove
  \begin{align*}
    \gquadruple{\emptyset}
      {((m_1, a_1, \pc), \infty, \bot)}
      {((m_2, a_2, \pc), \infty, \bot)}
      {((m_1, a_1, \pc, c), \infty, \bot)}
      {((m_2, a_2, \pc, c), \infty, \bot)}
  \end{align*}

  \noindent We will use all quadruples of states of this exact form as an invariant.
  We define \[
    I \triangleq \{ \ \textrm{all} \ \bigl(
      ((m_1, a_1, \pc), \infty, \bot),
      ((m_2, a_2, \pc), \infty, \bot),
      ((m_1, a_1, \pc, c), \infty, \bot),
      ((m_2, a_2, \pc, c), \infty, \bot)
    \bigr) \ \}
  \]
  and we use the \textsc{Invariant} rule to obtain $I$ as a (guarded) hypothesis.
  \begin{align*}
    \gquadruple{I}
      {((m_1, a_1, \pc), \infty, \bot)}
      {((m_2, a_2, \pc), \infty, \bot)}
      {((m_1, a_1, \pc, c), \infty, \bot)}
      {((m_2, a_2, \pc, c), \infty, \bot)}
  \end{align*}
  We note that when using the deductive system in Rocq, we do not have to explicitly write down the definition of $I$ (because it just consists of all states of the same form as the one currently being in focus).
  From there, we can start making progress by inspecting possible execution steps.
  We perform a case analysis on the instruction $P(\pc)$ currently being executed.

  The cases of $\kywd{add}$ and $\kywd{load}$ are very similar, but $\kywd{load}$ is slightly more involved because loading an address modifies the
  microarchitectural state.
  We therefore only detail the case of $\kywd{load}$.
  Supposing $P(\pc) = \pload{x_1}{x_2}$, we start by applying the rule \textsc{C-Step}.
  Importantly, this enables us to assume that the leakage is the same for both contract states.
  By definition of the always-mispredict contract, loads leak the loaded
  address, so we get to assume $a_1(x_2) = a_2(x_2)$.
  \begin{align*}
    &\gquadruple{I}
      {((m_1, a_1, \pc), \infty, \bot)}
      {((m_2, a_2, \pc), \infty, \bot)}
      {((m_1, a_1, \pc, c), \infty, \bot)}
      {((m_2, a_2, \pc, c), \infty, \bot)}\\
    \pcomment{\textsc{C-Step'}} & \textrm{Suppose} \ a_1(x_2) = a_2(x_2) = \mi{adr} \ \pcomment{1}\\
    &\gquadruple{I}
      {((m_1, a_1[x_1 \mapsto m_1(\mi{adr})], \pc + 1), \infty, \bot)}
      {((m_2, a_2[x_1 \mapsto m_2(\mi{adr})], \pc + 1), \infty, \bot)}
      {((m_1, a_1, \pc, c), \infty, \bot)}
      {((m_2, a_2, \pc, c), \infty, \bot)}
  \end{align*}
  Now, we can use the rule \textsc{H-Step} to get access to the hypothesis $I$.
  To do so, we have to prove that the leakage is the same for both hardware states.
  Here the content of the cache is leaked, which is not a problem because we started from the same cache $c$ in both hardware states.
  Further, the contract guarantees that the loaded addresses must be the same address $\mathit{adr}$ (assumption $(1)$), so the content of the cache remains the same in both hardware states even after loading.
  Therefore, the invariant $I$ is restored after executing the load, and we can conclude using \textsc{Cycle}.
  \begin{align*}
    \pcomment{\textsc{H-Step}} & c = c \ \wedge\\
    &\quadruple{I}
      {((m_1, a_1[\ldots], \pc + 1), \infty, \bot)}
      {((m_2, a_2[\ldots], \pc + 1), \infty, \bot)}
      {((m_1, a_1[x_1 \mapsto m_1(a_1(x_2))], \pc, a_1(x_2) \cons c), \infty, \bot)}
      {((m_2, a_2[x_1 \mapsto m_2(a_2(x_2))], \pc, a_2(x_2) \cons c), \infty, \bot)}\\
    \pcomment{By (1)}
    &\quadruple{I}
      {((m_1, a_1[\ldots], \pc + 1), \infty, \bot)}
      {((m_2, a_2[\ldots], \pc + 1), \infty, \bot)}
      {((m_1, a_1[x_1 \mapsto m_1(\underline{\mi{adr}})], \pc, \underline{\mi{adr}} \cons c), \infty, \bot)}
      {((m_2, a_2[x_1 \mapsto m_2(\underline{\mi{adr}})], \pc, \underline{\mi{adr}} \cons c), \infty, \bot)}\\
    \pcomment{\textsc{Cycle}} & (
      (\ldots, \infty, \bot),
      (\ldots, \infty, \bot),
      ((\ldots, \mi{adr} \cons c), \infty, \bot),
      ((\ldots, \mi{adr} \cons c), \infty, \bot)
    ) \in I
  \end{align*}

  \newcommand{\teq}{\,\overset{?}{=}\,}

  Next, we cover the case where $P(\pc) = \pjz{x}{l}$. This is the most interesting case
  since we have to reason about the different ways in which the CPU might
  speculate. We first take a contract step.
  \begin{align*}
    &\gquadruple{I}
      {((m_1, a_1, \pc), \infty, \bot)}
      {((m_2, a_2, \pc), \infty, \bot)}
      {((m_1, a_1, \pc, c), \infty, \bot)}
      {((m_2, a_2, \pc, c), \infty, \bot)}\\
    \pcomment{\textsc{C-Step'}} & \textrm{Suppose} \ (a_1(x) \teq 0) = (a_2(x) \teq 0) \ \pcomment{2}\\
    &\gquadruple{I}
      {((m_1, a_1, \pc_\bot), w, (m_1, a_1, \pc_\top))}
      {((m_2, a_2, \pc_\bot), w, (m_1, a_2, \pc_\top))}
      {((m_1, a_1, \pc, c), \infty, \bot)}
      {((m_2, a_2, \pc, c), \infty, \bot)}
  \end{align*}

  In this step, the always-mispredict contract jumps to the incorrect location first,
  saves the correct future state, and reveals the positivity of the register $x$.
  In particular, we get to assume that $(a_1(x) \teq 0) = (a_2(x) \teq 0)$.
  Thanks to this assumption, we can also deduce that the two contract states jump to the same (incorrect) location.
  Importantly, the two contract states are then in \emph{speculation mode} for $w$
  steps, where $w$ is the maximal speculation window.
  
  We now have to reason about what happens on the hardware side:
  the branch predictor will be queried, resulting in a jump to either the
  correct or the incorrect next location.
  Since our branch predictor depends only on the program counter,
  both hardware states will jump to the same (correct or incorrect) location.
  We denote the correct and the incorrect locations after executing $\pjz{x}{l}$ with $\pc_\top$ and $\pc_\bot$.
  
  The easier case is when the hardware, as the contract, also mispredicts.
  In this case, 
  both the contract and the hardware side are then temporarily executing the same (incorrect) code. Therefore, $w$ steps of straightforward lockstep reasoning are sufficient to skip-through the misprediction, re-establish the invariant $I$, and conclude.

  The more challenging case is when the hardware correctly
  predicts the outcome of the branch.
  Contrary to the previous cases, this one critically requires non-lockstep reasoning as the contract and the hardware are temporarily executing different instructions.
  Starting from the configuration where the two contract states are
  positioned at the same incorrect location, the two hardware states
  jump to the same \emph{correct} location after using $\textsc{H-Step}$.
  \begin{align*}
    \pcomment{\textsc{H-step} + (2)}
    &\gquadruple{I}
      {(\ldots, w, (m_1, a_1, \pc_\top))}
      {(\ldots, w, (m_2, a_2, \pc_\top))}
      {\hspace{-8pt}((m_1, a_1, \pc_\top, c), w, (\texttt{true}, (m_1, a_1, \pc_\bot)))}
      {\hspace{-8pt}((m_2, a_2, \pc_\top, c), w, (\texttt{true}, (m_2, a_2, \pc_\bot)))}
  \end{align*}
  Because the contract leaked the positivity of the register (assumption (2)), we know that the speculation will eventually be resolved in the same way for both hardware states (in particular, the flag is set to \texttt{true} in both).
  From there, we exploit the asynchronicity of our proof system to skip through $w$ steps
  on the contract side. After these $w$ steps, the two contract states roll back.
  \begin{align*}
    \pcomment{$w$ times \textsc{C-Step}}&\gquadruple{I}
      {((m_1, a_1, \pc_\top), \infty, \bot)}
      {((m_2, a_2, \pc_\top), \infty, \bot)}
      {\hspace{-8pt}((m_1, a_1, \pc_\top, c), w, (\texttt{true}, (m_1, a_1, \pc_\bot)))}
      {\hspace{-8pt}((m_2, a_2, \pc_\top, c), w, (\texttt{true}, (m_2, a_2, \pc_\bot)))}
  \end{align*}
  Finally, since both the contract and the hardware side are now re-aligned at the same program counter, we can conclude by $w - 1$ steps of lockstep reasoning,
  followed by an asynchronous hardware step to finalize the resolution of the speculation and restore the relational invariant.

  \begin{align*}
    \pcomment{$w - 1$ steps}&\gquadruple{I}
      {((m'_1, a'_1, \pc'), \infty, \bot)}
      {((m'_2, a'_2, \pc'), \infty, \bot)}
      {\hspace{-8pt}((m'_1, a'_1, \pc', c'), 0, (\texttt{true}, (m_1, a_1, \pc_\bot)))}
      {\hspace{-8pt}((m'_2, a'_2, \pc', c'), 0, (\texttt{true}, (m_2, a_2, \pc_\bot)))}\\
    \pcomment{\textsc{H-Step}}&\quadruple{I}
      {((m'_1, a'_1, \pc'), \infty, \bot)}
      {((m'_2, a'_2, \pc'), \infty, \bot)}
      {((m'_1, a'_1, \pc', c'), \infty, \bot)}
      {((m'_2, a'_2, \pc', c'), \infty, \bot)}\\
    \pcomment{\textsc{Cycle}}&
      (
        (\ldots, \infty, \bot)
        (\ldots, \infty, \bot)
        (\ldots, \infty, \bot)
        (\ldots, \infty, \bot)
      ) \in I
  \end{align*}

  We note that the first $w$ steps actually hide a bit of non-lockstep reasoning to cover the case of branches.
  Indeed, even though the hardware states cannot do nested speculation during these $w$ steps, the contract states still can. In that case it suffices to, again, skip through the misspeculation by repeated application of \textsc{C-Step}.
  In Rocq, these redundant phases of the proof can be conveniently factored as lemmas, and composed as needed. Indeed, the compositional nature of parameterized coinduction allows to prove a quintuple once, and reuse it in the derivation of another quintuple.

\subsection{Proof of the Sequential Contract}

Another key idea in the CT-SPEC contract is to
abstract away the fact that modern processors execute instructions \emph{out-of-order} 
and instead use a \emph{sequential} execution semantics.
To validate the correctness of this abstraction, we add simple out-of-order execution to our CPU model.
\arthur{The role of the contract is not so clear anymore with the new formulation.}

\paragraph*{Hardware Semantics and Contract}
We enrich the vanilla hardware semantics with a one-instruction buffer. Concretely, hardware states are pairs $(s, b)$, where $s$ is a vanilla hardware state, and $b$ is a buffer which can either be empty ($b = \bot$) or contain one instruction.
When the buffer is empty, instead of executing the current instruction at $s.\mi{pc}$ (call it $i_1$), the CPU can decide to instead store $i_1$ in the buffer and execute the instruction at $s.\mi{pc} + 1$ first (call it $i_2$).
Importantly, not all instructions can be executed out-of-order without
affecting the semantics of the executed programs.
We define a predicate $\mathit{delayable}(i_1, i_2)$ to describe when $i_1$ can be safely delayed.
We note that branches are not delayable (i.e., $\mathit{delayable}((\textbf{beqz} \cdot, \cdot), i_2)$ is false regardless of $i_2$).
When the buffer is full, the instruction in the buffer is immediately executed, leaving the buffer empty again.
On the contract side, the semantics is straightforward: contract states are composed of an architectural state $\sigma$,
and instructions are executed sequentially by following the ISA semantics.
The predicted leakage is the same as for the always-mispredict contract (loads leak the loaded address, and branches leak the result of the branch condition).

\paragraph*{Contract Satisfaction Proof}

We prove contract satisfaction by deriving the following quintuple, which requires initial states to have the same pc, cache, and buffer: \[
  \gquadruple{\emptyset}
      {(m_1, a_1, \pc)}
      {(m_2, a_2, \pc)}
      {((m_1, a_1, \pc, c),\bot)}
      {((m_2, a_2, \pc, c), \bot)}
\]
As in the previous example, we use all quadruples of states with this exact form as an invariant: \[
  I \triangleq \{ \ \textrm{all} \ ((m_1, a_1, \pc), (m_2, a_2, \pc), ((m_1, a_1, \pc, c),\bot), ((m_2, a_2, \pc, c),\bot)) \ \}
\]
After fixing $I$ as a guarded hypothesis using the \textsc{Invariant} rule, we first do a case analysis on whether the execution of the current instruction is delayed or not.
If not, the contract and the hardware proceeds in lockstep.
The interesting case is when an instruction is delayed. In this case,
the contract executes the instruction at $\mi{pc}$ first, while the hardware executes the instruction at $\mi{pc} + 1$ first.
In the proof, we therefore first use the rule \textsc{C-Step'} twice in order to pre-shot the leakage of the next \emph{two} execution steps.
We can then use \textsc{H-Step} twice to
simulate the execution of the instruction at $\mi{pc} + 1$ followed by the instruction in the buffer.
After this, the buffer is empty and the invariant is restored, enabling us to conclude the proof by \textsc{Cycle}.
Justifying the equality of hardware leakage after each \textsc{H-Step} and proving that the invariant is ultimately restored relies on 
two additional observations:
(1) swapping two delayable instructions preserves the leakage allowed by the contract, and (2) swapping the execution of delayable instructions preserves the semantics.
These two properties are proven as separate lemmas about the $\mi{delayable}$ predicate.

  \section{Up-To Techniques for Relative Bisimilarity}
  \label{sec:uptos}

  \subsection{Exploiting Symmetries and Transitivity in Relative Bisimulation Proofs}

  During a relative trace equality proof, it can be useful to replace one of the four states with another one in order to simplify the proof or to avoid redundant reasoning.
  For example, when proving $\den{s_1}_\sw = \den{s_2}_\sw \implies \den{h_1}_\hw = \den{h_2}_\hw$, we can always replace the contract state $s_1$ (resp. $s_2$) with another contract state $s'_1$ (resp. $s'_2$) such that $\den{s_i}_\sw = \den{s'_i}_\sw$.
  Similarly, by symmetry of equality, we can always swap $s_1$ and $s_2$, or $h_1$ and $h_2$.
  More subtly, we can also exploit transitivity of implication to split proofs.
  For example, if we have to prove $\den{s_1}_\sw = \den{s_2}_\sw \implies \den{h_1}_\hw = \den{h_2}_\hw$, and if we additionally know that $\den{s_1}_\sw = \den{s_2}_\sw \implies \den{s'_1}_\sw = \den{s'_2}_\sw$, we can instead prove $\den{s'_1}_\sw = \den{s'_2}_\sw \implies \den{h_1}_\hw = \den{h_2}_\hw$.
  Intuitively, this can be understood as replacing
  a pair of contract states with two states that ``leak less''.
  Again, the same observations applies to hardware states (i.e., we can replace a pair of hardware states with two states that ``leak more'').

  An interesting question is whether these \emph{tricks} transfer to proof quintuples. Since quintuples with an empty hypothesis are equivalent to relative trace equality (Theorem \ref{thm:quintuples-sound-complete}), it is straightforward to reformulate the above observations as reasoning rules when the hypothesis is $\emptyset$.
  It is less obvious that such principles also apply mid-proof, i.e., when the guarded hypothesis is \emph{not} empty.
  In fact, it not always the case!
  As an example, in the following figure, the rule on the left (which exploits transitivity to replace the hardware states with states that leak more) is easily proven sound, while its generalization on the right is unsound (we show why in \Cref{subsec:transitive-reasoning}).
  \begin{mathpar}
    \inferrule{{\gquadruple{\emptyset}{s_1}{s_2}{h'_1}{h'_2}}\\\gquadruple{\emptyset}{h'_1}{h'_2}{h_1}{h_2}}{{\gquadruple{\emptyset}{s_1}{s_2}{h_1}{h_2}}}
    \qquad \inferrule{{\gquadruple{R}{s_1}{s_2}{h'_1}{h'_2}}\\\gquadruple{\emptyset}{h'_1}{h'_2}{h_1}{h_2}}{{\gquadruple{R}{s_1}{s_2}{h_1}{h_2}}}
  \end{mathpar}

  This is unfortunate, because applying these principles mid-proof can, in some cases, dramatically simplify proofs.
  In particular, using such rules when the hypothesis is unguarded would enable using the \textsc{Cycle} rule even if the current states are not exactly in the hypothesis $R$, but \emph{almost} in it up to symmetries or up to transitivity.
  This would allow us to work with simpler relational invariants, and to better decompose proofs.

  \subsection{Compatible Up-To Techniques}

  In a standard proof by coinduction, we prove a statement $X \subseteq \nu F$ by identifying a coinduction hypothesis $I$ such that $X \subseteq I$ and $I \subseteq F(I)$.
  Often, we find $I$'s such that the second condition is not
  satisfied (i.e., $I \not\subseteq F(I)$),
  but it would be satisfied if we were considering a
  slightly larger $I$ on the right. That is, for a certain
  function $f$ (which given a set, returns a larger set),
  we can easily prove $I \subseteq F(f(I))$.
  It turns out that for $f$'s that satisfy a certain \emph{compatibility criterion}, reasoning up-to $f$-enlargements of $I$ is sound!
  This gives a powerful reasoning principle often referred to as \emph{coinduction up-to $f$} \cite{wayup}.

  \begin{definition}[Compatibility]
    Let $f$ and $F$ be two monotone functions over subsets of a set $U$.
    We say that $f$ is compatible with $F$ if $\forall X. \ f(F(X)) \subseteq F(f(X))$
  \end{definition}

  \begin{theorem}[Coinduction Up-to]
    Suppose $f$ is compatible with $F$, then \[
      (\exists I. \ X \subseteq I \wedge I \subseteq F(f(I))) \implies X \subseteq \nu F
    \]
  \end{theorem}

  The idea of coinduction up-to was first explored in the specific context of proofs by bisimulation \cite{SANGIORGI_1998}, and more recently extended in several ways to the general case of proofs by (parameterized) coinduction \cite{lattices_and_upto, paco, wayup, gpaco}.
  Building on these works (and in particular on the ``companion approach'' of Pous \cite{wayup}, and the GPaco library \cite{gpaco}), we can extend our deductive system with the following rule enabling the exploitation of any function $f$ compatible with $\texttt{rbisimF}$:
  \begin{mathpar}
    \inferrule[Up-To]{
      f \textrm{ is compatible with } \texttt{rbisimF}\\
      (s_1, s_2, h_1, h_2) \in f\Bigl(\lambda s'_1, s'_2, h'_1, h'_2. \ \ngquadruple{R}{s'_1}{s'_2}{h'_1}{h'_2}\Bigr)
    }{{\ngquadruple{R}{s_1}{s_2}{h_1}{h_2}}}
  \end{mathpar}

  Note that the absence of a frame around the hypothesis $R$ indicates that the rule \textsc{Up-to} can be used regardless of whether it is guarded or not, but guarded hypotheses remain guarded after applying the rule.
  We omit the soundness proof, which requires a slight adaption of the definition of quintuples (Definition~\ref{def:quintuples}), and we refer the readers to \cite{paco, wayup,gpaco} (and to our Rocq development) for details regarding the combination of up-to techniques and Paco-style proofs.

  Using the newly introduced \textsc{Up-to} rule, we obtain a systematic method to soundly integrate (and combine!) the reasoning principles discussed in the previous subsection.
  It suffices to encode each principle as a compatible function $f$.
  Although these compatibility proofs are
  fairly technical (they required us to prove a dozen auxiliary lemmas about \texttt{rbisimF}), the soundness of each principle then immediately follows from \textsc{Up-to}.
  We show in the next two subsections that all the principles discussed so far can be derived in this way, with the exception of the principle for augmenting hardware leakage. This last principle requires a subtle modification in order to be sound.

  \subsection{Reasoning Up To Symmetries and Equivalences}
  \label{subsec:symmetriesandequivalences}

  Reasoning up to symmetries and leakage-equivalent states is always sound, even mid-proof. This is reflected by the following set of rules. We write $\simeq$ for equivalence of leakage (i.e., trace equality).

  \begin{mathpar}
    \inferrule[C-Swap]{{\ngquadruple{R}{s_2}{s_1}{h_1}{h_2}}}{{\ngquadruple{R}{s_1}{s_2}{h_1}{h_2}}} \qquad
    \inferrule[H-Swap]{{\ngquadruple{R}{s_1}{s_2}{h_2}{h_1}}}{{\ngquadruple{R}{s_1}{s_2}{h_1}{h_2}}} \qquad
    \inferrule[C-Leak-Eq]{s_1 \simeq s'_1 \\ {\ngquadruple{R}{s'_1}{s_2}{h_1}{h_2}}}{{\ngquadruple{R}{s_1}{s_2}{h_1}{h_2}}} \qquad
    \inferrule[H-Leak-Eq]{h_1 \simeq h'_1 \\ {\ngquadruple{R}{s_1}{s_2}{h'_1}{h_2}}}{{\ngquadruple{R}{s_1}{s_2}{h_1}{h_2}}}
  \end{mathpar}

  The four rules are derived from \textsc{Up-To} by choosing an appropriate compatible function.
  For example, the compatible function implementing \textsc{C-Swap} is: \[
    f_{\textsc{C-Swap}}(R) \triangleq \{ \ (s_1, s_2, h_1, h_2) \mid (s_2, s_1, h_1, h_2) \in R \ \}
  \]
  The functions implementing all other rules can be found in \ifthenelse{\boolean{fullversion}}{Section~\ref{sec:up-to-functions} of the appendix}{the full version of the paper~\cite{fullVersion}}.

  We note that by combining \textsc{Swap} rules and \textsc{Leak-Eq} rules, we can replace any of the four states with a leakage-equivalent one.
  Further, since equality of leakage is equivalent to bisimilarity,
  the side-goals $s_1 \simeq s'_1$ and $h_1 \simeq h'_2$ can also be proven by parameterized coinduction.

  \subsection{Transitive Reasoning}
  \label{subsec:transitive-reasoning}
  Composing quintuples horizontally is sound, even mid-proof.
  This allows decomposing a proof by artificially reducing (resp. augmenting) leakage at the contract level (resp. hardware level) as reflected by the following two rules.

  \begin{mathpar}
    \inferrule[Reduce-Contract-Leakage]{{\gquadruple{\emptyset}{s_1}{s_2}{s'_1}{s'_2}}\\{\ngquadruple{R}{s'_1}{s'_2}{h_1}{h_2}}}{{\ngquadruple{R}{s_1}{s_2}{h_1}{h_2}}} \qquad
    \inferrule[Augment-Hardware-Leakage]{{\ngquadruple{R}{s_1}{s_2}{h'_1}{h'_2}}\\{\sgquadruple{\emptyset}{h'_1}{h'_2}{h_1}{h_2}}}{{\ngquadruple{R}{s_1}{s_2}{h_1}{h_2}}}
  \end{mathpar}

  Importantly, the two rules are not exactly symmetrical: on the hardware side, we can only compose with a \mi{lockstep} variant of our quintuples on the right, defined as \[
    \sgquadruple{R}{h'_1}{h'_2}{h_1}{h_2} \triangleq (h'_1, h'_2, h_1, h_2) \in G_{\texttt{rbisimF}_\texttt{lockstep}}(R)
  \]
  Without this restriction, the rule \textsc{Augment-Hardware-Leakage} is unsound. Indeed, consider the following scenario:
  \begin{mathpar}
    \leak_\hw(h_1) \neq \leak_\hw(h_2) \qquad
    \leak_\hw(h'_1) = \leak_\hw(h'_2) \quad \leak_\hw(\trans_\hw(h'_1)) \neq \leak_\hw(\trans_\hw(h'_2))
  \end{mathpar}
  Under these assumptions, for any contract state $s$, we can easily prove the following first two quintuples, and yet disprove the third one: \begin{align*}
    \gquadruple{\top}{s}{s}{h'_1}{h'_2}
    \qquad
    \gquadruple{\emptyset}{h'_1}{h'_2}{h_1}{h_2}
    \qquad
    \neg\Bigl(\gquadruple{\top}{s}{s}{h_1}{h_2}\Bigr)
  \end{align*}
  The first quintuple is proven using \textsc{H-Step}
  (exploiting $\leak(h'_1) = \leak(h'_2)$), and then \textsc{Cycle}.
  The second quintuple is proven using \textsc{C-step} followed by \textsc{C-Leak} (exploiting $\leak_\hw(\trans_\hw(h'_1)) \neq \leak_\hw(\trans_\hw(h'_2))$).
  Even though we proved these two quintuples, the third one does not hold.
  This is because the two contract traces are equal by definition (so we can never use \textsc{C-Leak} to conclude), and because $\leak_\hw(h_1) \neq \leak_\hw(h_2)$ (so we cannot use \textsc{H-Step} to conclude either).
  The main issue in this counterexample is that the second quintuple (involving $h'_1, h'_2, h_1, h_2$) only holds because $h'_1$ and $h'_2$ have different traces, which we were able to discover using non-lockstep reasoning.
  However, the quintuple carries no information
  about $h_1$ and $h_2$, nor about $s_1$ and $s_2$.
  Working with lockstep relative bisimilarity instead fixes the problem: it forces a tighter connection between the four states.

  \section{Discussion: Limitations and Future Work}
  \label{sec:discussion}

  \paragraph{Additional Contract Features}
  We used the proof system to prove some of the key features of the original CT-SPEC contract (the one that uses the always-mispredict technique)~\cite{hwswcontracts}, but there is still some work to be done to prove the entire contract correct.
  Most notably, our hardware semantics uses a simplified scheduler that does not feature a fetch-execute-retire pipeline.
  Furthermore, we are yet to combine the always-mispredict feature with the out-of-order scheduler.
  While the proof will require a non-negligible amount of work, judging from our experience using the proof system, we are confident that the proof will be feasible.
  Additionally, we have not considered nested speculation in this paper.
  Nested speculation is usually modeled 
  by an explicit stack at the contract level.
  At the hardware level, though,
  speculation is typically implemented
  in a less structured way (e.g, using buffers).
  As a result, reasoning about nested speculation requires subtle invariants linking 
  stacks with lower-level data-structures.
  Although our deductive system is complete, finding and expressing such invariants might be difficult in practice.
  To facilitate the process, an idea would be to extend our deductive system with built-in support for reasoning about a
  \emph{logical} view of the hardware
  states rather than directly
  against low-level data-structures (e.g. taking inspiration of \emph{ghost states} in modern separation logics \cite{iris}).
  This would certainly enable working with simpler invariants connecting contract and hardware states.
  \jana{Arthur, could you add a sentence what makes nested speculation tricky?}
  \arthur{I did my best to explain how ghost states could help}

  \paragraph{Towards a Program Logic}
  Our proof system simplifies the process of mechanizing contract satisfaction proofs. Still, it requires explicit reasoning about concrete hardware and contract states. Conducting proofs with the help of a proof assistant such as Rocq significantly helps with the associated bookkeeping. However, for larger proofs and more complicated examples, reasoning implicitly about a more abstract set of states is desirable.
  To achieve this, an idea would be to enable reasoning about 4-ary predicates
  expressed in a carefully curated logic, rather than relying on the logic
  of the proof assistant to reason explicitly about quadruples of concrete states.
  This would certainly enable more modular proofs, in the style of what can be
  achieved using modern mechanized separation logics such as \cite{iris} or \cite{simuliris}.

  \paragraph{Nondeterminism}
  In this paper, we assumed that both the contract and the hardware semantics are deterministic.
  This follows the assumptions made in recent works on hardware-software contracts~\cite{hwswcontracts, speculationatfault,Leave,mohr2024synthesizing,mohr2025synthesis}.
  Adding support for nondeterminism would be a non-trivial research challenge.
  The first question would be how to define contract satisfaction in that setting.
  The current definition (Definition~\ref{def:hwswcontracts}) would not be adequate anymore, as traces might change either because of leakage or because of nondeterminism. \arthur{I don't think it's a correctness issue, it's just that Def 2.1 gives a very strong notion of contract satisfaction in the nondet setting. Probably too strong to prove it in practice. But in a sense, it's just "too" correct.}
  One option could be to use a $\forall\exists$ quantifier alternation, inspired by how the noninterference property can be generalized to nondeterministic settings~\cite{gennoninterference}.
  In the presence of a quantifier alternation, we would have to develop new reasoning principles to resolve existential quantifiers.
  \arthur{I am not sure about the last two sentences: seems like this would be a very very weak notion of contract satisfaction. Maybe we should just stop at "can be generalized to nondeterministic seetings"}

  \paragraph{Relative Bisimulation for SNI}

  The development of our deductive system was motivated by contract satisfaction proofs, but it could find further applications.
  A particularly interesting one is proving Speculative Non-Interference (SNI)~\cite{spectector}, another instance of relative trace equality.
  SNI requires a program to not leak more than what is leaked 
  through sequential, non-speculative execution of the CPU.
  Combining proofs of contract satisfaction and SNI guarantees the absence of information leakage when running a program~\cite{hwswcontracts}: 
  first prove that a CPU satisfies a contract, and then prove that specific programs satisfy SNI with respect to that contract.
  Both proofs can in principle be developed within our deductive system.

  \section{Related Work}

  \paragraph{Contract Validation}
  There is a growing interest in validating CPUs against hardware-software contracts. Validation efforts either employ testing (mostly of black-box CPUs)~\cite{revizor, hideandseek, speculationatfault, amulet, hideandseek}, or they model-check white-box designs~\cite{Leave,shadowlogic}.
  When model-checking hardware-software contracts, the biggest challenge is finding invariants that enable the model checker to find a proof. Existing works tackle this challenge with inductive invariant learning (e.g., LeaVe~\cite{Leave}) or by requiring additional user input (e.g., shadow logic~\cite{shadowlogic}).
  While having the advantage of automation, model-checking techniques might abort (in the case of bounded model checking) or not terminate, leaving the user with little information as to why the proof fails. Additionally, model checking is vulnerable to logical and implementation bugs.
  We are the first to propose a foundational proof system for generating fully mechanized proofs for contract satisfaction.

  \paragraph*{Secure Compilation}
  Our proof system builds on the rich tradition of simulation-based proofs for security.
  4-ary reasoning is particularly prominent in the context of secure compilation to reason about the preservation of constant-timeness~\cite{ctjasmin, barthe2021structured, ctsim, compcertsec}, noninterference~\cite{securitymultithreaded, BARTHE200735}, and, generally, k-safety hyperproperties~\cite{securecomphyper,niopt}.
  Many of these works have also been formalized in a proof assistant.
  Recently, secure compilation w.r.t. constant-timeness and noninterference has been extended to speculative semantics as well~\cite{specconsttime, Wall2024SNIPSE, exorcisingspectres}.
  There are some key differences between our work and secure compilation.
  When proving that a compiler preserves noninterference, 
  high-level objects (source programs) and low-level objects (target programs) are \emph{structurally} very different, either because the compiler is optimizing, or because the source and target languages are different.
  Further, there is a strict dependency between high and low level (i.e., we only reason about target programs obtained by compiling a source program).
  In the case of contracts, there is no compiler between the higher level (the contract) and the lower level (the hardware).
  The program is the same across the four compared executions, but it is executed under two different semantics.
  Our proof system take advantages of these observations, and lets the developer establish an artificial relation between the contract and the hardware for the purpose of the proof.

  \paragraph*{4-ary Reasoning Beyond Secure Compilation}
  Apart from secure compilation, there exist additional security properties that require 4-ary reasoning.
  Example are speculative noninterference~\cite{spectector,hwswcontracts} and \emph{information-preserving refinement}. for the latter, different simulation techniques have been chained to obtain a formal proof spanning the entire hardware-software stack~\cite{AthalyeCKTZ24, knox}.
  Another example is \emph{relative security}, which compares the information flows occurring under different semantics~\cite{DongolGPW24}.
  Hardware-software contract satisfaction can be seen as an instance of relative security.
  Dongol et al~\cite{DongolGPW24} introduce a simulation-based proof technique for relative security.
  In contrast to our setting, they rely on proving that a collection of conditions hold for a simulation relation provided up-front. Our proof system lets us construct the simulation relation on the fly.
  Another distinction is that they use timers to bound their vanilla traces (corresponding to our contract traces), ensuring they do not execute infinitely without making progress. In our framework, this is handled implicitly through the use of coinductive-inductive predicates (the inductive part implicitly encodes the progress assumptions).
  We designed our notion of relative bisimilarity with contract satisfaction in mind, but we expect that it could be applicable to relative security as well.
  \jana{Arthur: I am being a bit bold here, saying that HW-SW contract satisfaction is an instance of relative security and that our relative bisimilarity notion might is probably applicable to that property as well. Check if you're ok with that.}

  \paragraph*{Parameterized Coinduction}

  Our work is strongly rooted in recent developments in the field of interactive proofs by coinduction. In particular, our proof technique leverages the framework of parameterized coinduction, introduced by Hur et al. \cite{paco}.
  Our Rocq formalization is based on their implementation of parameterized coinduction, Paco \cite{pacolib}.
  To prove the soundness of the techniques we presented in Section 5, we also leverage the notion of up-to techniques for coinductive proofs.
  In the paper introducing Paco, Hur et al. already identified that up-to techniques
  can be combined with parameterized coinduction.
  Later on, Pous developed a more general and systematic approach for the principled development and integration of up-to techniques via the so-called \emph{companion approach}~\cite{wayup}.
  More recently, Zakowski et al. showed that the approach of Pous can also be
  integrated in Paco via generalized parameterized coinduction~\cite{gpaco}.
  The extensions we developed in Section 5 rely on their implementation of generalized parameterized coinduction, which is now integrated in the Paco library~\cite{pacolib}.

  \paragraph{Interaction Trees}

  Interaction Trees (ITrees) are a general purpose coinductive structure for representing impure computations \cite{itree}.
  ITrees are mechanized in Rocq, and come associated with several notions of behavioral equivalences, including ones specifically designed to prove security properties such as non-interference \cite{itree_sec}.
  As our proof quintuples, these equivalences are defined as coinductive relations, and come associated with high-level reasoning principles facilitating proofs.
  However, to the best of our knowledge, ITrees do not feature reasoning principles for direct 4-ary reasoning in the style of our deductive system.
  \jana{Remove that paragraph to shorten? After all, our focus is not on coinductive proofs anymore... Alternatively, remove the next paragraph?}

  \paragraph*{Coinductive Proofs for Hyperproperties}

  Recently, Correnson and Finkbeiner applied parameterized coinduction
  to the verification of temporal hyperproperties \cite{hyco}.
  Their mechanized framework, HyCo, supports the verification of a large class of hyperproperties.
  However, contract satisfaction is out of the scope their proof framework.
  In particular, HyCo does not support asynchronous reasoning in the style of our
  coinductive-inductive relation.

  \section*{Data Availability}

  The Rocq development accompanying this paper is available on Zenodo at the following address: \begin{center}
    \href{https://doi.org/10.5281/zenodo.19054166}{https://doi.org/10.5281/zenodo.19054166}
  \end{center}

  \section*{Acknowledgements}
  This work was supported by the Deutsche Forschungsgemeinschaft (DFG, German Research Foundation) under Germany’s Excellence Strategy - EXC 2092 CASA - 390781972 and by the European Research Council (ERC) Grant HYPER (No. 101055412). Views and opinions expressed are however those of the authors only and do not necessarily reflect those of the European Union or the European Research Council Executive Agency. Neither the European Union nor the granting authority can be held responsible for them. A. Correnson carried out this work as a member of the Saarbr\"ucken Graduate School of Computer Science.

  \bibliographystyle{ACM-Reference-Format}
  \bibliography{references}

  \ifthenelse{\boolean{fullversion}}{

  \appendix

  \section{Full Architectural and Hardware Semantics}
  \label{app:fullSemantics}
  
  Figures~\ref{fig:vanilla-arch}~and~\ref{fig:vanilla-hard} contain the full vanilla architectural and hardware semantics.

  \begin{figure}[ht!]
    \begin{mathpar}
    \inferrule[Add]
    { v = a(x_2) + k}
    {
      \tup{m, a, \pc, \tadd{x_1}{x_2}{k}} \archsingleStep{} \tup{m, a[x_1 \mapsto v], \pc + 1}
    }

    \inferrule[Load]
    { v = m(a(x_2)) }
    {
      \tup{m, a, \pc, \pload{x_1}{x_2}} \archsingleStep{} \tup{m, a[x_1 \mapsto v], \pc + 1}
    }

    \inferrule[Branch]
    {\pc' = (a(x) \overset{?}{=} 0) ~?~ l : (\pc + 1)}
    {
      \tup{m, a, \pc, \pjz{x}{l}} \archsingleStep{} \tup{m, a, \pc'}
    }
    \end{mathpar}
      \caption{Vanilla architectural semantics.}\label{fig:vanilla-arch}
  \end{figure}%

  \begin{figure}[ht!]
    \begin{mathpar}
      \inferrule[Load]
      {
        (\sigma, \pload{x_1}{x_2}) \archsingleStep{} \sigma'
      }
      {
        \tup{\sigma, c, \pload{x_1}{x_2}} \hwsingleStep{c} \tup{\sigma', a(x_2) \cons c}
      }

      \inferrule[Other]
      {
        (\sigma, i) \archsingleStep{} \sigma'
      }
      {
        \tup{\sigma, c, i} \hwsingleStep{c} \tup{\sigma', c}
      }
      \end{mathpar}
      \caption{Vanilla hardware semantics.}
      \label{fig:vanilla-hard}
  \end{figure}

  \section{Example 2: Out-of-Order Scheduler}
  \label{app:OOOexample}
  This section presents the proof of our second example, the out-of-order scheduler in more detail. The full formalization and proof can be found in our accompanying Rocq development.
  This example gives a proof that out-of-order execution does not impact the leakage behavior if it preserves functional correctness. This is not trivial, as the out-of-order execution impacts the cache, which is observable by the attacker.
  We define a simple out-of-order scheduler that can decide to delay the execution of an instruction for one step by putting it on a heap. If there is an instruction on the heap, it is the one executed next.
  \jana{Revisit this text once we know what we put in the main paper.}

  \begin{figure}
\begin{mathpar}

    \inferrule[Execute]
    {
      \Sigma(s.\pc) = \mathit{execute} \\
      (s, P(s.\pc)) \hwsingleStep{o} s'
    }
    {
      \tup{s, \bot} \hwstep{o} \tup{s', \bot}
    }

    \inferrule[Execute Heap]
    {
      (s, i) \hwsingleStep{o} s'
    }
    {
      \tup{s, i} \hwstep{o} \tup{s'[\pc \mapsto s.\pc], \bot}
    }

    \inferrule[Delay]
    {
      P(s.\pc) = i_1 \\
      \Sigma(s.\pc) = \mathit{delay} \\
      P(s.\pc + 1) = i_2 \\
      \mathit{delayable}(i_1, i_2) \\
      s[\pc \mapsto s.\pc + 1] \hwsingleStep{o} s'
    }
    {
      \tup{s, \bot} \hwstep{o} \tup{s', i_1}
    }
    \end{mathpar}
    \centering
    \begin{align*}
      \mathit{delayable}(i_1, i_2) &= 
            i_1 \neq \pjz{\cdot}{\cdot} \land
      \mathit{des}(i_1) \neq des(i_2) \land  
      \mathit{des}(i_1) \neq src(i_2) \land 
      \mathit{src}(i_1) \neq des(i_2)
    \end{align*}
    \begin{minipage}{0.45\linewidth}
      \begin{alignat*}{2}
      &\mathit{src}(\tadd{x_1}{x_2}{k}) &&=  x_2 \\
      &\mathit{src}(\pload{x_1}{x_2}) &&=  x_2 \\
      &\mathit{src}(\pjz{x}{l}) &&= x
    \end{alignat*}
    \end{minipage}%/
    \begin{minipage}{0.45\linewidth}
      \begin{alignat*}{2}
      &\mathit{des}(\tadd{x_1}{x_2}{k}) &&= x_1  \\
      &\mathit{des}(\pload{x_1}{x_2}) &&=  x_1  \\
      &\mathit{des}(\pjz{x}{l}) &&= s.\pc
    \end{alignat*}
    \end{minipage}%
  \caption{Hardware semantics with out-of-order execution.}\label{fig:ooosem}
  \end{figure}

  \paragraph{Out-of-Order Hardware Semantics}
  The semantics of the out-of-order scheduler is given in \Cref{fig:ooosem}. It uses a scheduling function $\Sigma$, which depends on the program counter.
  The functionality of the out-of-order scheduler is implemented in rule \textsc{Delay}.
  If the heap is empty, the scheduler can decide to delay the current instruction and execute the next one. This is only possible if the instruction is \emph{delayable}. This predicate is encoded with a predicate that checks if one instruction impacts the execution of the other. The predicate is defined as expected: two instructions can be swapped if the components one instruction writes to are not read or written by the other instruction.
  Note that the combination of the predicate and the scheduling decision can result in the execution getting stuck.
  In our formalization, we only consider valid schedulers — that is, ones that return delay only  when the hardware state is delayable, but not necessarily every time it is.
  %  since we assume infinite traces, we define the semantics to loop in the last state in that case.

  \begin{figure}
    \begin{mathpar}   
      \inferrule[Step]
      { 
        \sigma \archsingleStep{} \sigma' \\
        \cstepstyle{o} =
        {
          \begin{cases}
            \cstepstyle{\sigma.a(x_2)} & \text{if } P(\sigma.\pc) = \pload{x_1}{x_2}\\
            \cstepstyle{\sigma.a(x) \,\overset{?}{=}\, 0} & \text{if } P(\sigma.\pc) = \pjz{x}{l} \\
            \cstepstyle{\bot} & \text{otherwise}
          \end{cases}
        }
      }
      {
        \sigma  \amstep{o}  \sigma'
      } 
    \end{mathpar}
    \caption{Sequential contract.}
    \label{fig:seq-contract}
  \end{figure}

  \paragraph{Sequential Contract}
  The out-of-order execution can be abstracted away with a leakage contract that follows a simple in-order semantics that we formalize in \Cref{fig:seq-contract}. The contract just exposes accessed addresses and the result of branching decisions. It is similar to the sequential contract CT-SEQ described previously~\cite{hwswcontracts}. 

  \paragraph{Contract Satisfaction Proof}
  
To establish that the sequential contract is satisfied by our simple out-of-order execution hardware model, we rely on the soundness of our proof system for relative trace equality (Theorem \ref{thm:quintuples-sound-complete}).
Accordingly, for any initial memories and register assignments $m_1,m_2,a_1, a_2$, for any initial cache $c$ and initial program pointer $\pc$, and for any valid scheduler $\Sigma$, it suffices to show that:
  \begin{align*}
    \gquadruple{\emptyset}
      {(m_1, a_1, \pc)}
      {(m_2, a_2, \pc)}
      {((m_1, a_1, \pc, c),\bot)}
      {((m_2, a_2, \pc, c), \bot)}
  \end{align*}

As in the always-mispredict example, we apply the \textsc{Invariant} rule to establish the relational invariant needed for the proof:
  \begin{align*}
    I = \{ \ (s_1, s_2, (h_1, \bot), (h_2, \bot)) \mid \
    s_i.m = h_i.m \wedge
    s_i.a = h_i.a \wedge
    s_i.\pc = h_i.\pc \wedge
    s_1.\pc = s_2.\pc \wedge
    h_1.c = h_2.c \
    \}
  \end{align*}
  After applying the \textsc{Invariant} rule, we now have $I$ as a (guarded) hypothesis.
  \begin{align*}
    \gquadruple{I}
      {(m_1, a_1, \pc)}
      {(m_2, a_2, \pc)}
      {((m_1, a_1, \pc, c),\bot)}
      {((m_2, a_2, \pc, c), \bot)}
  \end{align*}

  We perform a case analysis on the results of scheduler $\Sigma(\pc)$. When the scheduler returns $\mathit{execute}$, the hardware step is fully synchronized with the contract. In this case, we only need to show that after both the contract and hardware take a step, the relational invariant $I$ still holds.
    \begin{align*}
    &\gquadruple{I}
      {(m_1, a_1, \pc)}
      {(m_2, a_2, \pc)}
      {((m_1, a_1, \pc, c),\bot)}
      {((m_2, a_2, \pc, c), \bot)}\\
    \pcomment{\textsc{C-Step', H-Step}} 
    &\quadruple{I}
      {\trans_\sw(m_1, a_1, \pc)}
      {\trans_\sw(m_2, a_2, \pc)}
      {\trans_\hw((m_1, a_1, \pc, c),\bot)}
      {\trans_\hw((m_2, a_2, \pc, c), \bot)} \\
    \pcomment{\textsc{Cycle}} &
      (\trans_\sw(\ldots),
      \trans_\sw(\ldots),
      \trans_\hw(\ldots,\bot),
      \trans_\hw(\ldots, \bot)) \in I
  \end{align*}
We note that in order to use \textsc{H-Step}, we first need to prove equality of leakage: \[
  \leak_\hw((m_1, a_1, \pc, c), \bot) = \leak_\hw((m_2, a_2, \pc, c), \bot)
\]
According to our cache-based attacker model, only the cache is leaked. Since the cache is the same in both hardware states at this point, the equality trivially holds (the leakage is $c$ on both sides).
We also note that to use the rule \textsc{Cycle},
we need to show that the invariant $I$ still holds after
executing the current instruction.
Fortunately, using the rule \textsc{C-Step'} generates the assumption $\leak_\sw(m_1, a_1, \pc, c) = \leak_\sw(m_1, a_1, \pc, c)$. This assumption is sufficient to prove that the invariant is restored.
The only non-trivial case is when a load instruction is executed. In this case the cache changes, but since the contract leaks the loaded address, the cache remains the same in both hardware states and the invariant $I$ is restored.

When the scheduler returns $\mathit{delay}$, there is a mismatch between the contract and the hardware execution.
The hardware executes the instruction at $\mi{pc} + 1$ first, and then the instruction at $\mi{pc}$.
The contract, however, executes instructions in order ($\mi{pc}$, then $\mi{pc} + 1$).
To handle this misalignment, the trick is to first
simulate \emph{two} contract steps using \textsc{C-Step'}, and then simulate \emph{two} hardware steps using \textsc{H-Step}.
\begin{align*}
    &\gquadruple{I}
      {(m_1, a_1, \pc)}
      {(m_2, a_2, \pc)}
      {((m_1, a_1, \pc, c),\bot)}
      {((m_2, a_2, \pc, c), \bot)}\\
    \pcomment{\textsc{C-Step'} $\times$ 2} &\quadruple{I}
      {\trans_\sw(\trans_\sw(m_1, a_1, \pc))}
      {\trans_\sw(\trans_\sw(m_2, a_2, \pc))}
      {((m_1, a_1, \pc, c),\bot)}
      {((m_2, a_2, \pc, c), \bot)}\\
    \pcomment{H-Step $\times$ 2} &\quadruple{I}
      {\trans_\sw(\trans_\sw(m_1, a_1, \pc))}
      {\trans_\sw(\trans_\sw(m_2, a_2, \pc))}
      {\trans_\hw(\trans_\hw((m_1, a_1, \pc, c),\bot))}
      {\trans_\hw(\trans_\hw((m_2, a_2, \pc, c), \bot))}\\
  \pcomment{\textsc{Cycle}}
      &(\trans_\sw(\trans_\sw(\ldots)),
      \trans_\sw(\trans_\sw(\ldots)),
      \trans_\hw(\trans_\hw(\ldots)),
      \trans_\hw(\trans_\hw(\ldots))) \in I
  \end{align*}

To justify these five proof steps, we have a quite a few things to unpack: we need to justify the applicability of \textsc{H-Step}, and we need to prove that the invariant $I$ was indeed restored.

We start by noting that the applications of \textsc{C-Step'} generate two equality assumptions: \begin{align*}
  \leak_\sw(\ldots) = \leak_\sw(\ldots) \wedge
  \leak_\sw(\trans_\sw(\ldots)) = \leak_\sw(\trans_\sw(\ldots))
\end{align*}
Using \textsc{H-Step} two times requires us to prove the two following equality: \begin{align*}
  \leak_\hw((\ldots, c),\bot) = 
  \leak_\hw((\ldots, c), \bot) \wedge
  \leak_\hw(\trans_\hw(\ldots, \bot)) = 
  \leak_\hw(\trans_\hw(\ldots, \bot))
\end{align*}
To justify that the two applications of \textsc{H-Step} were valid, we can therefore prove the following implication:
\begin{align*}
  &\leak_\sw(\ldots) = \leak_\sw(\ldots)
  \wedge
  \leak_\sw(\trans_\sw(\ldots)) = \leak_\sw(\trans_\sw(\ldots)) \implies\\
  &\leak_\hw((\ldots, c),\bot) = 
  \leak_\hw((\ldots, c), \bot) \wedge
  \leak_\hw(\trans_\hw(\ldots, \bot)) = 
  \leak_\hw(\trans_\hw(\ldots, \bot))
\end{align*}

Since a valid scheduler $\Sigma$ only returns $\mathit{delay}$ if the state is $\mathit{delayable}$, this also implies that instruction $P(\pc)$ cannot be a branch. Thus, we have 6 combination of instructions to consider.
As the cases are very similar, we only focus on the case of two successive loads as a representative example. I.e., we suppose $P(\pc) = i_1 = \pload{x_1}{x_2}$ and $P(\pc + 1) = i_2 = \pload{x_1'}{x_2'}$.

First, let us examine how our implication can be proved in this case. Since both instructions are loads,
the contract leaks the loaded addresses.
On the hardware side, the cache is leaked. After unfolding the definitions of $\leak_\sw$ and $\leak_\hw$, our implication therefore becomes:
\begin{align*}
  &a_1(x_2) = a_2(x_2) \wedge 
  \trans_\sw(\ldots).a(x_2') =
  \trans_\sw(\ldots).a(x_2') \implies \\
  &c = c \wedge
  \trans_\hw(\ldots, \bot).a_1(x_2) =
  \trans_\hw(\ldots, \bot).a_2(x_2)
\end{align*}
The equality $c = c$ is trivially true, but we still need to show \[
  \trans_\hw(\ldots, \bot).a(x_2) = \trans_\hw(\ldots, \bot).a(x_2)
\]
Unfortunately, the equality assumptions generated by the contract do not immediately match:
we only know $a_1(x_2) = a_2(x_2)$.
However, we know that the two successive loads satisfy the $\mi{delayable}$ predicate.
By definition (see Figure~\ref{fig:ooosem}), this means that reordering their executions does not affect the
values they read.
In particular, $a_1(x_2)$ must be equal to $\trans_\hw(\ldots, \bot).a_1(x_2)$,
and $a_2(x_2)$ must be equal to $\trans_\hw(\ldots, \bot).a_2(x_2)$.
This concludes the leakage equality proof.

We still need to show that the invariant was restored after 
the two contract steps and the two hardware steps.
In the case of two successive loads, it is easy to verify that swapping their execution order does not impact the final state. Formally, assuming $i_1$ and $i_2$ are loads, for any microarchitectural states $h, h_{i_1}, h_{i_2}, h'$ can show:
\begin{align*}
  (h, i_1) \hwsingleStep{} h_{i_1} \wedge
  (h_{i_1}, i_2) \hwsingleStep{} h' \iff
  (h, i_2) \hwsingleStep{} h_{i_2} \wedge
  (h_{i_2}, i_1) \hwsingleStep{} h'
\end{align*}

Since we can swap the execution of these two load instructions without changing the hardware state, the contract states and the architectural states of the hardware states are synchronized after two steps. At this point, the reorder buffer is empty again, and the invariant $I$ is restored.

\section{Up-to Functions}
\label{sec:up-to-functions}

Section~\ref{sec:uptos} extends our deductive
system with several up-to techniques to simplify proofs.
The soundness of each up-to technique is justified by picking a compatible function $f$ and instantiating the rule \textsc{Up-To}.
The following table explicitly describes which function is chosen to justify each up-to technique.

\begin{center}
\renewcommand{\arraystretch}{1.5}
\small
\begin{tabular}{|c|l|}
  \hline
  \textbf{Rule} & \textbf{Corresponding Compatible Function}
  \\\hline
  \textsc{C-Swap} & $f(R) \triangleq \{ \ (s_1, s_2, h_1, h_2) \mid (s_2, s_1, h_1, h_2) \in R \ \}$
  \\\hline
  \textsc{H-Swap} & $f(R) \triangleq \{ \ (s_1, s_2, h_1, h_2) \mid (s_1, s_2, h_2, h_1) \in R \ \}$
  \\\hline
  \textsc{C-Leak-Eq} & $f(R) \triangleq \{ \ (s_1, s_2, h_1, h_2) \mid (s_1, s'_1) \in \texttt{bisim} \wedge (s'_1, s_2, h_1, h_2) \in R \ \}$
  \\\hline
  \textsc{H-Leak-Eq} & $f(R) \triangleq \{ \ (s_1, s_2, h_1, h_2) \mid (h_1, h'_1) \in \texttt{bisim} \wedge (s_1, s_2, h_1, h'_1) \in R \ \}$
  \\\hline
  \textsc{Reduce-C-Leakage} & $f(R) \triangleq \{ \ (s_1, s_2, h_1, h_2) \mid (s_1, s_2, s'_1, s'_2)  \in \texttt{rbisim} \wedge (s_1', s'_2, h_1, h_2) \in R \ \}$
  \\\hline
  \textsc{Augment-H-Leakage} & $f(R) \triangleq \{ \ (s_1, s_2, h_1, h_2) \mid (h'_1, h'_2, h_1, h_2) \in \texttt{rbisim}_\texttt{lockstep} \wedge (s_1, s_2, h'_1, h'_2) \in R \ \}$
  \\\hline
\end{tabular}
\end{center}

}{}

\end{document}